\documentclass[12pt]{article}

\usepackage{natbib}
\usepackage{setspace}
\usepackage[left=1in,top=1in,right=1in,bottom=1in]{geometry}
\usepackage{graphicx}
\usepackage{amsthm}
\usepackage{amssymb}
\usepackage{amsmath}

\usepackage{version}
\usepackage[dvipsnames]{xcolor}
\usepackage[normalem]{ulem}
\usepackage{multirow}
\usepackage{multicol}
\usepackage{epstopdf}
\usepackage{xcolor}
\usepackage{subcaption}
\usepackage{floatrow}
\usepackage{enumitem}
\usepackage{tikz}
\usetikzlibrary{arrows,arrows.meta,decorations,decorations.pathreplacing,calc,matrix}

\definecolor{Red}{rgb}{1,0,0}
\definecolor{Blue}{rgb}{0,0,1}
\definecolor{Green}{rgb}{0,1,0}
\definecolor{magenta}{rgb}{1,0,.6}
\definecolor{lightblue}{rgb}{0,.5,1}
\definecolor{lightpurple}{rgb}{.6,.4,1}
\definecolor{gold}{rgb}{.6,.5,0}
\definecolor{orange}{rgb}{1,0.4,0}
\definecolor{hotpink}{rgb}{1,0,0.5}
\definecolor{newcolor2}{rgb}{.5,.3,.5}
\definecolor{newcolor}{rgb}{0,.3,1}
\definecolor{newcolor3}{rgb}{1,0,.35}
\definecolor{darkgreen1}{rgb}{0, .35, 0}
\definecolor{darkgreen}{rgb}{0, .6, 0}
\definecolor{darkred}{rgb}{.75,0,0}
\definecolor{lightgrey}{rgb}{.7,.7,.7}

\definecolor{clemson-orange}{RGB}{234,106,32}
\definecolor{chicago-maroon}{RGB}{128,0,0}
\definecolor{northwestern-purple}{RGB}{82,0,99}
\definecolor{cornell-red}{RGB}{179,27,27}
\definecolor{sauder-green}{RGB}{171,180,0}
\definecolor{lawngreen}{RGB}{0,250,154}

\usepackage[unicode=true,
 bookmarks=true,bookmarksnumbered=true,bookmarksopen=true,bookmarksopenlevel=4,
 breaklinks=false,pdfborder={0 0 1},backref=false,colorlinks=true]
 {hyperref}
\hypersetup{ pdfborderstyle={},urlcolor=chicago-maroon,linkcolor=chicago-maroon,citecolor=chicago-maroon}

\DeclareMathOperator*{\argmax}{\textrm\upshape{arg\,max}}

\usepackage{times}
\newtheorem{theorem}{Theorem}

\newtheorem{assumption}{Assumption}

\newtheorem{corollary}{Corollary}

\newtheorem{definition}{Definition}
\newtheorem{example}{Example}

\newtheorem{lemma}{Lemma}

\newtheorem{proposition}{Proposition}

\newtheorem{remark}{Remark}

\usepackage[nameinlink]{cleveref}
\crefname{assumption}{Assumption}{Assumptions}
\crefname{lemma}{Lemma}{Lemmas}
\crefname{theorem}{Theorem}{Theorems}
\crefname{corollary}{Corollary}{Corollaries}
\crefname{proposition}{Proposition}{Propositions}
\crefname{claim}{Claim}{Claims}
\crefname{procedure}{Procedure}{Procedures}
\crefname{algorithm}{Algorithm}{Algorithms}
\crefname{figure}{Figure}{Figures}
\crefname{remark}{Remark}{Remarks}
\crefname{section}{Section}{Sections}
\crefname{procedure}{Procedure}{Procedures}
\crefname{example}{Example}{Examples}
\crefname{definition}{Definition}{Definitions}
\crefname{table}{Table}{Tables}
\crefname{equation}{}{}
\crefname{enumi}{}{}
\crefname{conjecture}{Conjecture}{Conjectures}
\crefname{step}{Step}{Steps}
\crefname{appendix}{Appendix}{Appendices}
\crefname{footnote}{Footnote}{Footnotes}

\begin{document}

\title{Nested Search\footnote{This paper grew out of a conversation with Eddie Dekel. I am grateful to him and to Asher Wolinsky for continuous guidance and support. I thank Arjada Bardhi, Vivek Bhattacharya,  Hersh Chopra, Piotr Dworczak, Stephan Lauermann, Shuo Liu, Yijun Liu, Suraj Malladi, Yoon Sang Moon,  
Ilya Morozov, Alessandro Pavan, Harry Pei, Harrison Ridland, Fanqi Shi, Marciano Siniscalchi, Bruno Strulovici, Siqi Zheng, Shusheng Zhong, Yangfan Zhou for helpful comments. I am especially thankful to Wojciech Olszewski for continuous discussion through all stages of the project. All errors are my own.}
\author{Yutong Zhang\thanks{Department of Economics, Northwestern University. Email: zhangyutong2017@u.northwestern.edu}\\\textit{Preliminary and Incomplete}}}
\date{\today}
\maketitle

\begin{abstract}
I introduce and study a nested search problem modeled as a tree structure that generalizes \cite{weitzman1979optimal} in two ways: (1) search progresses incrementally, reflecting real-life scenarios where agents gradually acquire information about the prizes; and (2) the realization of prizes can be correlated, capturing similarities among them. I derive the optimal policy, which takes the form of an index solution. I apply this result to study monopolistic competition in a market with two stages of product inspection. My application illustrates that regulations on drip pricing lower equilibrium price and raise consumer surplus.

\end{abstract}


\thispagestyle{empty}



\let \markeverypar \everypar
\newtoks \everypar
\everypar \markeverypar
\markeverypar{\the \everypar \looseness=-2\relax}
\newpage
\setcounter{page}{1}

\section{Introduction}
\cite{weitzman1979optimal} studies a dynamic decision-making problem in which an agent, faced with a set of boxes containing unknown and independently realized prizes, sequentially searches through them at some cost. I study a nested search model that generalizes this problem in two ways. First, search progresses incrementally, reflecting real-life scenarios where agents gradually acquire information about the prizes. Second, the realization of prizes may be correlated, capturing similarities among them.

I employ a tree structure that captures both features simultaneously. Before search begins, the agent is positioned at the tree’s root, and each terminal node represents a prize. Except for the root, each node is associated with a random variable and a cost; by paying that cost to visit the node, the agent observes the realized value of its random variable. The prize at any terminal node is determined by the realization of the random variables of its ancestors and itself. At each step, the agent decides either (i) which uninspected child node of any visited node to inspect next or (ii) whether to stop searching and claim a box whose inspections are complete (if there is any). The agent cannot inspect a child node without first inspecting all its ancestors.

This tree structure accommodates the first, incremental search, feature by gradually revealing prize information as the agent moves from the root toward a terminal node, so the prize’s true realization becomes known only after all inspection stages are finished. The second feature, correlated payoffs, is  built into the tree: terminal nodes sharing a longer path from the root exhibit greater similarity.

In my problem, an index policy maximizes the agent's expected payoff. It assigns an index to every node except the root, where each node’s index is a random variable depending solely on the realized signals of its ancestors and on the distributions and search costs of the subtree rooted at that node. The policy is to visit the node with the highest index, provided this index exceeds the highest prize among fully inspected boxes.

From a theoretical point of view, my work contributes to the directed search literature in the following three aspects: First, it enriches the standard search model by considering multi-dimensional attributes of each prize, as each ancestor of a terminal point can be interpreted as one of its attributes. Second, the tree structure allows me to introduce correlation to the model and flexibly capture different levels of similarities  among boxes. Lastly, it is a middle ground between sequential search, where the agent decides which box to inspect at each step, pays the associated cost, observes the revealed prize, and then chooses the next box, and simultaneous search, where the agent chooses a subset of the boxes to inspect, and the realization of all prizes is revealed only after paying all the search cost. In my model, only partial information of the prize is revealed to the agent to guide future search decision.

From a practical point of view, my model naturally applies to numerous real-world contexts that cannot be captured by \cite{weitzman1979optimal} and numerous applied papers based on it: 
\begin{enumerate}
    \item R\&D: A pharmaceutical company aims to develop a new drug for sale. The research process unfolds in a nested, step-by-step manner guided by scientific protocols—early studies inform subsequent development stages, culminating in a final, market-ready product. In this framework, the terminal nodes represent the potential final drugs the company could bring to market. Progressing through the tree requires costly experiments at each stage to evaluate the viability of different drug candidates.
    \item Housing market: A buyer is searching for a house to purchase. The buyer begins by inspecting various neighborhoods to gather preliminary information on available houses. Next, the search narrows to specific streets within a chosen neighborhood. Finally, individual houses are inspected in greater detail to make a final decision.
    \item Degree-seeking: A student pursuing a college degree begins by completing prerequisite courses common to multiple majors, gathering broad insights. Next, he  narrows his focus and explores advanced classes to learn about potentials of different degrees and finally decides on which degree to get.
    \item Online search: A buyer is looking for a mug on an e-commerce platform, such as Amazon. After entering keywords ``mug", initial information—such as prices, monthly sales, and ratings—appears on the screen. To learn more detail, including sizes, colors, the buyer must click on a product page and be redirected to a new page with additional information.\footnote{In this example, products with different designs are treated as distinct products.} The buyer can always exit the page and click on a different product page if he does not find the products match his preference.
\end{enumerate}

\paragraph{Results} The first part of the paper (Sections \ref{s2}-\ref{s3}) deals with the general model and characterizes the optimal search policy. The key result in this part is the optimality of an index policy (Theorem \ref{opt}). I define the indexes in Definition \ref{cap} and formally state the index policy in Definition \ref{defindex}. To prove its optimality, I proceed in three steps. I first prove an upper-bound on the agent's expected payoff for any policy (Lemma \ref{upperbound}). Then I show that this upper-bound is tight for the index policy (Lemma \ref{nonexposure}). Finally, I show that the index policy also maximizes the upper-bound (Lemma \ref{maxbound}), thus maximizing the agent's expected payoff. An important by-product of Theorem \ref{opt} is a discrete choice formulation, which proves useful for analyzing consumer surplus and demand in the application I study.

The second part of the paper applies the theoretical results to an application. In Section \ref{s4}, I study monopolistic competition in a market with two stages of product inspection. There is an infinite number of firms in the market and each firm offers one product for sale. The buyer can pay a positive cost to acquire partial information about his match value for a firm’s product and, based on that information, either incur an additional cost to learn the true realization of the match value or stop and visit another firm. I characterize the equilibrium price under scenarios where prices are revealed at the first stage (Proposition \ref{equi1}) or at the second stage (Proposition \ref{equi2}) and compare it with the benchmark case with a single stage of product inspection. I find that when prices are revealed at the first stage, the equilibrium price is always lower than that in the benchmark (Theorem \ref{lowerprice}). I identify two market forces that lead to this decrease. First, the option to stop after observing partial information about the match value makes it more worthwhile to sample another unvisited firm, which intensifies competition (Lemma \ref{option}); second, when only partial information is revealed at the first stage, the buyer responds more aggressively as he leaves the firm immediately after seeing an unfavorable signal, without bothering to check the remaining information, which could still turn out to be favorable (Lemma \ref{elasticity}). Furthermore, since the buyer also benefits from the option value generated by gradual resolution of uncertainty, the consumer surplus is higher than that of the benchmark case (Corollary \ref{highercs}). 

When prices are revealed at the second stage, the equilibrium price is higher than that when it is revealed at the first stage (Proposition \ref{stage2higher}). However, comparison of prices with the benchmark is ambiguous. I identify two markets forces that work in opposite directions. On the one hand, similar as before, the option to stop halfway intensifies competition and tends to drag down the equilibrium price; on the other hand, for a buyer who enters the second stage, the firm knows that he must have observed a favorable signal at the previous stage and thus has an incentive to raise the price. Consumer surplus is likewise ambiguous as the gain from the option value may be fully offset (or even outweighed) by the higher equilibrium price. 

I then endogenize the timing of price revelation as part of the firms' strategy and study the equilibrium price revelation policy (when to reveal and what to reveal) in markets with or without drip pricing regulation. Under the regulation, if a firm discloses a price at stage 1, the price disclosed at stage 2 must be weakly lower. I find that in unregulated markets, since the price revealed at stage 1 is only a cheap talk message, any equilibrium is outcome equivalent to one in which each firm reveals its price at stage 2, and this price coincides with the equilibrium price that arises when the revelation policy is exogenously fixed to disclosure at stage 2 (Proposition \ref{unique2}). In a regulated market, by contrast, the regulation enables firms to indeed commit to the price revealed at stage 1 and this price coincides with the equilibrium price that arises when the revelation policy is exogenously fixed to disclosure at stage 1 (Proposition \ref{unique1}). Since the market price is always lower when prices are revealed early (Proposition \ref{stage2higher}), consumers indeed benefit from such regulation (Corollary \ref{effregulation}).

\paragraph{Literature review}
This paper is related to several strands of literature. I discuss them separately.

\bigskip
\noindent\textit{Pandora's problem and branching bandit}

My work is closely related to research on Pandora’s problem, originally introduced by \cite{weitzman1979optimal}. Since then, various extensions have been studied, including nonobligatory inspection in \cite{doval2018whether}, simultaneous search in \cite{chade2006simultaneous}, general utility functions in \cite{olszewski2015more} and search with learning in \cite{adam2001learning}. My work generalizes this model by allowing for gradual resolution of uncertainty and correlation embedded in a tree structure. For the first feature, \cite{kleinberg2016descending} and \cite{bowers2024matching} consider Pandora's problem with multiple stages of inspection for each box and establish the optimality of Gittins index policy following a proof in the spirit of \cite{weber1992gittins}. For the second feature, the computer science literature also relaxes the independence assumption, but it typically combines both correlation and learning about the distribution of the prizes. The focus is also more on the performance of different algorithms, rather than economic applications. For further reference, see Section 2.4 in \cite{beyhaghi2024recent}.




My model is a special case of a branching bandit problem, for which the optimality of an index solution was established long ago by \cite{weiss1988branching}. Several other economics papers who are also special cases of it include \cite{fershtman2021searching} and \cite{greminger2022optimal}. The mostly closely related paper is \cite{keller2003branching}, who also formulate their problem using a tree structure and give a non-constructive proof of the optimality of an index policy. In contrast, I provide a constructive proof in the spirit of \cite{weber1992gittins} that yields explicitly characterizations of the agent's expected payoff and choice behavior, which cannot be obtained from \cite{keller2003branching}.

\bigskip
\noindent\textit{Price competition}

\cite{weitzman1979optimal} is typically used in IO theory literature to study price competition when consumers need to incur positive search cost to inspect each firm's product match value or price (or both). A large literature studies how search frictions shape equilibrium pricing and market power. The most seminal papers are \cite{wolinsky1986true} and \cite{anderson1999pricing}, who study price competition in the presence of search cost and product differentiation. My model differs in that I consider markets with gradual resolution of uncertainty and study how it affects price competition and consumer surplus, which cannot be studied in their model with one shot resolution of uncertainty.

\bigskip
\noindent\textit{Drip pricing}

One strand of the theoretical literature shows that drip pricing can arise from behavioral bias (e.g., \cite{gabaix2006shrouded}). Another strand shows that hidden fees and obfuscation can arise even with fully rational consumers in the presence of search cost. For example, in models of add-on pricing and obfuscation, firms may strategically delay or complicate the revelation of relevant price components, which softens price competition and raises profit (e.g., \cite{ellison2012search}, \cite{ellison2005model}). My work falls in the second strand but differs from existing work in that I study a two-stage search model and explicitly model drip pricing policy as what and when to reveal the price. Also, the economic force that drives my result is new: regulation forbids a higher price at stage 2, therefore giving firms commitment power over stage 1 price. Competition thus happens at stage 1, when agents are not locked in by a high realization of match value.


\paragraph{Outline} The rest of the paper is organized as follows: Section \ref{s2} describes the model setup; Section \ref{s3} defines an index policy and proves its optimality; Sections \ref{s4}  explores an application of nested search, using the index solution as a tool; Section \ref{s5} provides a discussion of the results; Section \ref{s6} concludes. Supporting analysis and omitted proofs are relegated to the Appendix.

\section{Model}\label{s2}
I lay out the model in Sections \ref{s21}-\ref{s23}.  In Section \ref{s25}, I discuss and justify modeling choices using real-life examples.
\subsection{(Nested) tree structure}\label{s21}
An agent is presented with a set of boxes, each containing an unknown prize. He can take at most one prize or obtain a fixed outside option of 0. Unlike the setup in \cite{weitzman1979optimal}, the prize distributions here may be correlated. This correlation, interpreted as  similarities among boxes, is modeled by a tree structure where each terminal node represents a box. Let \(R\) denote the root node. Its child nodes are \(S_1, S_2, \cdots, S_n\), representing the first child node, the second child node and so on. For \(S_i\), its child nodes are \(S_{i,1},S_{i,2},\cdots,S_{i,k}\), representing the first child node of \(S_i\), the second child node of \(S_i\) and so on. The tree is recursively defined in this way. Two nodes that share the same parent node are called siblings. A node is called a terminal node if it has no child nodes. Let \(\mathcal{T}\) denote the set of terminal nodes, i.e., boxes. I use \(e_{i_1,\cdots,i_k}\) to denote the edge \((S_{i_1,i_2,\cdots,i_{k-1}},S_{i_1,i_2,\cdots,i_{k-1},i_{k}})\). There is a random variable \(X_{i_1,i_2,\cdots,i_{k}}\), distributed according to \(F_{i_1,\cdots,i_k}\), attached to each edge \(e_{i_1,\cdots,i_k}\), with \(x_{i_1,i_2,\cdots,i_{k}}\) representing its realization. The realization of the final prize for a box \(S_{i_1,\cdots,i_{n}}\) is determined by all the edges connecting it to the root through the function \(v^{i_1,\cdots,i_n}(x_{i_1},x_{i_1,i_2},\cdots,x_{i_1,\cdots,i_n})\). Therefore, the more common edges shared by two boxes, the more similarity they have. Since \(v\) depends on the whole path, even for two boxes sharing many edges and thus realizations of \(X\), their true prizes  can be very different, even though they are indeed more correlated. For the node \(S_{i_1,\cdots,i_k}\), let \(\mathcal{D}_{i_1,\cdots,i_k}\) denote its descendants, including itself.

\subsection{Search}\label{s22}
The realization of each random variable is unknown ex-ante; the agent has to incur some search cost to learn its true realization. Specifically, the search cost associated with edge \(e_{i_1,\cdots,i_k}\) is represented by \(c_{i_1,i_2,\cdots,i_{k}}\) and is assumed to be a constant. The agent has to pay this cost to learn \(x_{i_1,\cdots,i_k}\). There is a one-to-one correspondence between the edge \(e_{i_1,\cdots,i_k}\) and the node \(S_{i_1,\cdots,i_k}\), so inspecting \(e_{i_1,\cdots,i_k}\) and inspecting node \(S_{i_1,\cdots,i_k}\) have the same meaning and I use them interchangably. 
When searching, the agent has to follow the order specified by the tree structure, i.e., the agent cannot inspect a node without first inspecting all its ancestors. However, at any point, he may choose not to inspect the child nodes of the node he just probed and jump to a previously inspected  node to examine its uninspected child node, or cease searching and select a prize among the boxes for which all inspection stages have been completed (if there is any). Importantly, the agent can claim a box only after he has completed all inspection stages, i.e., inspecting all the edges connecting it to the root. Search is with free recall, i.e., the agent can always go back to a previous node or box without paying additional cost. 

I assume the tree structure, the distribution of each random variable \(X_{i_1,\cdots,i_k}\), the search cost \(c_{i_1,\cdots,i_k}\) and the value function \(v^{i_1,\cdots,i_n}\) are known to the agent before search. Additionally, I assume that all correlations are captured by the tree structure, as stated below:
\begin{assumption}\label{independence}
Conditional on \(X_{i_1},\cdots,X_{i_1,\cdots,i_{n-1}}\), \(X_{i_1,\cdots,i_n}\) is independent of the random variables associated with all other nodes except its descendants.
\end{assumption}
This assumption requires that, given any node \(S_{i_1,\cdots,i_k}\), once conditioned on the values of the random variables corresponding to the edges connecting it to the root, the realization of \(X_{i_1,\cdots,i_k}\) is independent of those associated with \(S_{i_1,\cdots,i_k}\)'s  siblings and the siblings of its ancestors. In \cite{weitzman1979optimal} where every terminal node is a child of the root, Assumption \ref{independence} reduces to the classic independence assumption.
\begin{remark}
There are two ways to understand the function \(v\). First, one may interpret the random variables as attributes that contribute to the final payoff through the function \(v\), as in the hedonic pricing model commonly used to estimate real estate values. Alternatively, the random variables can be viewed as a sequence of informative signals about the true valuation. In this interpretation, for a terminal node \(S_{i_1,\cdots,i_n}\), the payoff is given by \(v^{i_1,\cdots,i_n} = x_{i_1,\cdots,i_n}\), while the other random variables along the path, although not entering \(v^{i_1,\cdots,i_n}\) directly, are correlated with \(X_{i_1,\cdots,i_n}\).

\end{remark}

\subsection{Policy}\label{s23}
At each stage, given the set of nodes the agent has inspected and the set of boxes for which all inspection stages are complete, a policy either inspects an uninspected node \(S_{j_1,\ldots,j_k}\) that is a child of some already-inspected node, or stops searching and takes the larger of the highest realized prize and the outside option.  Given a policy \(\chi\), let \(\mathcal{O}(\chi)\) denote the random set of boxes whose inspection stages are fully completed, and let \(\mathcal{S}(\chi)\) be the set of nodes inspected under \(\chi\).\footnote{\(\mathcal{S}(\chi)\) is a sufficient statistic for \(\mathcal{O}(\chi)\). I introduce \(\mathcal{O}(\chi)\) for ease of exposition.}

My goal is to identify the optimal policy that maximizes the agent's expected payoff:
\begin{align*}
    \mathbb{E}\left[\max \Biggl\{\max_{S_{i_1,\cdots,i_n} \in \mathcal{O}(\chi)}v^{i_1,\cdots,i_n}(X_{i_1},\cdots,X_{i_1,\cdots,i_n}),0\Biggr\} -\sum_{S_{j_1,\cdots,j_k} \in \mathcal{S}(\chi)} c_{j_1,\cdots,j_k} \right].
\end{align*}

\subsection{Discussion of modeling choices}\label{s25}
I discuss the key assumptions used throughout the paper.
\paragraph{Precedence constraint} I assume the agent must follow the search order specified by the tree structure. This assumption is especially appropriate when the search process has a natural logical progression. For instance, in many R\&D contexts—such as drug design—the research process typically proceeds in a structured, step-by-step manner dictated by scientific protocols. This reinforces the assumption that the search must follow a predetermined order. In the degree-seeking example, students must first complete prerequisite courses before enrolling in more advanced classes. Furthermore, in the context of online search, platforms can often impose a fixed inspection order on users, who are typically unable to skip steps during the process.

\paragraph{Obligatory inspection}
I assume an obligatory inspection procedure, meaning the agent cannot claim any box without completing all its inspection stages. This assumption is supported by both theoretical considerations and practical relevance. Theoretically, it is a standard assumption embedded in most search models; without it, as shown by \cite{doval2018whether}, the problem becomes prohibitively complex. Practically, the assumption aligns well with many real-world applications. It is particularly reasonable in R\&D settings where a firm must finalize every stage of product development before bringing a salable product to market. Degree-seeking is another example, where a student must complete all courses, both preliminary and advanced, to fulfill the requirements and obtain a degree. In online search environments, platforms or firms often structure the process so that the purchase decision is deferred to the very end until all inspection steps have been completed.

\paragraph{A priori known tree structure}
I assume that the tree structure is fully known to the agent prior to the search. This assumption is reasonable in several settings: in R\&D, where research process is often highly stylized; in degree-seeking contexts, where all course requirements are specified in advance; and in online search environments, where the sequence of steps is typically predetermined. Nonetheless, it is certainly of interest to consider an extension of the model in which the tree structure is random.

\section{Optimal policy}\label{s3}
In this section, I define an index policy and prove its optimality. In Section \ref{s31}, I calculate an index for each node (except for the root) and formalize the index solution. To prove its optimality in Section \ref{s32}, I first provide an upper-bound on the agent's expected payoff. Then I prove that under the index policy, the upper-bound is obtained with equality. Finally, I show that this upper-bound is maximized by the index policy. 

\subsection{Index policy}\label{s31}
Let \(N_0\) be the number of child nodes of the root and \(N_{i_1,\cdots,i_k}\) be the number of child nodes of \(S_{i_1,\cdots,i_k}\). For any terminal node \(S_{i_1,\cdots,i_n}\), set \(N_{i_1,\cdots,i_n}=1\).\footnote{Terminal nodes have no child nodes, but I set the number to 1 so that \(\kappa\) is well defined in Definition \ref{cap} for terminal nodes.} I define for each node in the tree (except for the root) two variables, capped value \(\kappa\) and index \(\sigma\). My definition generalizes Definition 4 in \cite{bowers2024matching} and incorporates it as a special case:
\begin{definition}\label{cap}
The index for the node \(S_{i_1,\cdots,i_k}\), which is denoted \(\sigma_{i_1,\cdots,i_k}\), is a function of \(X_{i_1},\cdots,X_{i_1,\cdots,i_{k-1}}\) that uniquely solves
\begin{align*}
    \mathbb{E}\left[\left(\max_{1 \leq l \leq N_{i_1,\cdots,i_k}} \kappa_{i_1,\cdots,i_{k},l}-\sigma_{i_1,\cdots,i_k}\right)^{+}\Big|X_{i_1},\cdots, X_{i_{k-1}}\right]=c_{i_1,\cdots,i_{k}},
\end{align*}
where the capped value is defined as \(\kappa_{i_1,\cdots,i_k}=\min \left\{\sigma_{i_1,\cdots,i_k},\max_{1 \leq l \leq N_{i_1,\cdots,i_k}} \kappa_{i_1,\cdots,i_k,l}\right\}\), and for a terminal node \(S_{i_1,\cdots,i_n}\), define  \(\sigma_{i_1,\cdots,i_n,1}=\kappa_{i_1,\cdots,i_n,1}=v^{i_1,\cdots,i_n}(x_{i_1},\cdots,x_{i_1,\cdots,i_n})\).
\end{definition}
To get some intuition on how the capped values and indexes are defined, note that for \(S_{i_1,\cdots,i_n} \in \mathcal{T}\), conditional on the realization of \(X_{i_1},\cdots,X_{i_1,\cdots,i_n}\), its index is exactly the reservation value defined in \cite{weitzman1979optimal}, which is the maximum value of the outside option that makes the agent indifferent between stopping immediately and inspecting the box. The indexes in my definition can be defined in a similar manner. To see this , Theorem 1 in \cite{kleinberg2016descending} says that Pandora's problem can be reformulated  into a discrete choice problem where there is no search cost and the payoff from any box is its capped value \(\kappa_{i_1,\cdots,i_n}=\min \left\{\sigma_{i_1,\cdots,i_n},v^{i_1,\cdots,i_n}(x_{i_1},\cdots,x_{i_n})\right\}\). Then, when deciding whether to inspect its parent node \(S_{i_1,\cdots,i_{n-1}}\) or not, the payoff from inspecting it comes from being able to claim a prize from one of its child nodes. So the maximum payoff is \(\max_{1 \leq l \leq N_{i_1,\cdots,i_{n-1}}} \left\{\kappa_{i_1,\cdots,i_{n-1},l}\right\}\) and the maximum value of the outside option that makes the agent indifferent between stopping immediately and inspecting \(S_{i_1,\cdots,i_n}\) is \(\sigma_{i_1,\cdots,i_{n-1}}\) defined in Definition \ref{cap}. The indexes and capped values are defined in such a recursive way. Following this definition, all indexes and capped values can be obtained by going backwards from the terminal nodes.

The index policy I consider is similar to the ones in \cite{weitzman1979optimal}, \cite{kleinberg2016descending} and \cite{bowers2024matching}. The formal definition is as follows:
\begin{definition}\label{defindex}
Given any \(\mathcal{S}\), the index policy \(\chi^*\) proceeds as follows: at each decision point, it inspects the child node (of an already inspected node) with the highest index \(\sigma\), provided that this index exceeds both the highest realized prize among the boxes in \(\mathcal{O}\) and the outside option of 0; otherwise, the policy takes the highest realized prize or the outside option (whichever is higher). Ties are broken randomly.
\end{definition}
It is worth noting that except for child nodes of the root, the indexes are random variables whose true values are only revealed after inspecting all their ancestors. During the search process, the agent updates the indexes based on past observations, compares them, and then inspects the node with the highest index or take the highest prize.

The main result of the first part of the paper is the optimality of the index policy:
\begin{theorem}\label{opt}
The index policy maximizes the agent's expected payoff.
\end{theorem}
The proof of the theorem comprises three key steps, each detailed in the next subsection.
\subsection{Optimality of the index policy}\label{s32}
For each node \(S_{i_1,\cdots,i_k}\), let \(\mathbb{I}_{i_1,\cdots,i_k}\) be the indicator that the agent inspects this node. For \(S_{i_1,\cdots,i_n} \in \mathcal{T}\), let \(\mathbb{A}_{i_1,\cdots,i_n}\) be the indicator that the agent claims it. The agent's expected payoff depends on both his inspection decision and his decision on which box to claim. However, for any policy, there exists an upper-bound on the agent's expected payoff that only involves \(\mathbb{A}\), as stated in the lemma below:
\begin{lemma}\label{upperbound}
The agent's expected payoff is bounded by the following:
\begin{align*}
\mathbb{E}\left[\sum_{l=1}^{N_0}\left(\sum_{S_{i_1,\cdots,i_n} \in \mathcal{T} \cap \mathcal{D}_l}\mathbb{A}_{i_1,\cdots,i_n}\right)\kappa_{l}\right].
\end{align*}
\end{lemma}
This lemma is reminiscent of Lemma 2 in \cite{bowers2024matching} and includes it as a special case. It establishes that the agent’s expected payoff is bounded above by the payoff from an auxiliary discrete choice problem. In this problem, for each child node \(S_l\) of the root, the agent decides whether to claim a box from \(\mathcal{D}_{l}\), considering \(\kappa_l\) as the payoff for doing so. This upper-bound  applies to all policies, but it is tight for the index policy, as stated in the following lemma:
\begin{lemma}\label{nonexposure}
    The index policy achieves the exact upper-bound in Lemma \ref{upperbound}.
\end{lemma}
Lemma \ref{nonexposure} does not establish the optimality of the index policy, as the bound in Lemma \ref{upperbound} depends on \(\mathbb{A}\), which can vary across different policies. There could exist a policy that fails to obtain the bound in Lemma \ref{upperbound} but still generates a higher payoff. I show that this is impossible as the upper-bound is maximized under the index policy in Lemma \ref{maxbound} below:

\begin{lemma}\label{maxbound}
The upper-bound in Lemma \ref{upperbound} is maximized by the index policy.   
\end{lemma}
Lemma \ref{upperbound}-\ref{maxbound} combine to establish the optimality of the index policy. For any other policy, the agent's expected payoff is bounded by the formula in Lemma \ref{upperbound}. Furthermore, the formula is maximized by the index policy. Lastly, the agent's expected payoff under the index policy attains this maximized value. This proves Theorem \ref{opt}.

\section{Application: pricing game with gradual resolution of uncertainty}\label{s4}

In this section, I apply theoretical results derived earlier to analyze the pricing game in a market featuring two stages of product inspection. I then focus on monopolistic competition (where there is an infinite number of firms) and compare the equilibrium price and consumer surplus to the benchmark case where total cost is the same, but all uncertainty is resolved in a single stage. Finally, I examine the impact of drip pricing regulation by solving for the equilibrium price in a modified pricing game that endogenizes firms’ price revelation policies.

\subsection{Setup}
There are \(n\) firms in the market, each offering a product for sale. The production cost is normalized to 0. There are \(L\) buyers who would like to purchase at most one product. Their outside option is normalized to 0. Let \(s=\frac{L}{n}>0\) denote the average market share. Each firm sets its price \(p_i\) simultaneously and it is unobservable before search. Each buyer's match value for firm \(i\)'s product is \(Z_i=X_i+Y_i\). Conditional on buying this product at price \(p_i\), the buyer's payoff is \(X_i+Y_i-p_i\). Both \(X_i\) and \(Y_i\) are unobservable before search. By paying cost \(c_X>0\), the buyer can learn a firm's \(X_i\). If additionally he pays \(c_Y>0\), he also learns \(Y_i\). For tractability, I assume  \(X_i\) and \(Y_i\) are i.i.d. across firms and buyers. 
\(X\) follows distribution \(F\), supported on an interval \([\underline{x},\overline{x}]\) with density \(f>0\). \(Y\) follows distribution \(G\), supported on an interval \([\underline{y},\overline{y}]\) with density \(g>0\). I assume \(\overline{x}+\overline{y}>0\); otherwise the problem is vacuous. Whenever a buyer is indifferent among firms, he breaks ties evenly. I study the symmetric pure strategy Nash equilibrium of this game that features active consumer search when (i) price is observed at stage 1; (ii) price is observed at stage 2.\footnote{There always exists an equilibrium in which each firm posts  a prohibitively high price and buyers never search.}\footnote{If the price is observed before stage 1 costlessly, then there is no pure strategy equilibrium, and the mixed strategy equilibrium is difficult to characterize, as discussed in \cite{armstrong2017ordered}.}

\subsection{Price revealed at stage 1}\label{s42}
In this subsection, I study the case where the price is revealed together with \(X\) at stage 1. Suppose all firms but one charges \(p^*\) and the remaining firm, say firm \(i\), charges \(p\). I assume that the buyer will hold the same belief about the prices charged by other firms even after observing a deviation from \(p^*\). To tell whether firm \(i\) has an incentive to deviate, one must determine its demand. There are many possible search trajectories that lead to a buyer finally purchasing firm \(i\)'s product. This challenge can be resolved by Lemma \ref{maxbound}, which gives a generalized version of the discrete choice theorem developed in \cite{choi2018consumer}. I slightly abuse notation by letting \(\kappa_l\) denote the lowest index of firm \(l\). Then \(\kappa_j= \min \{\sigma-p^*,X_j+r-p^*,X_j+Y_j-p^*\}\)  for \(j \neq i\) and \(\kappa_i=\min \{\sigma-p^*,X_i+r-p,X_i+Y_i-p\}\), where \(\mathbb{E}[(Y-r)^+]=c_Y\) and \(\mathbb{E}[(X+\min \{Y,r\}-\sigma)^+]=c_X\). Notice that the first index for firm \(i\) is \(\sigma-p^*\) because buyers do not observe its price before search and will hold a belief that firm \(i\) sets the putative equilibrium price. Firm \(i\)'s demand can be decomposed into two parts: first, when \(\kappa_i=\sigma-p^*\), once a buyer visits firm \(i\), he will complete both stages of inspection and purchase at firm \(i\), and this happens when \(\kappa_j<\sigma-p^*\) for all previously visited firm \(j\). Second, when \(\kappa_i<\sigma-p^*\), a buyer finally purchases at firm \(i\) if all other firms offer worse products, i.e., \(\kappa_j \leq \kappa_i<\sigma-p^*\), \(\forall j \neq i\). Define \(W_l=X_l+\min \{Y_l,r\}\) and let \(H\) be its distribution function. Since \(F\) and \(G\) each are supported on intervals with strictly positive densities, \(H\) is supported over \([\underline{x}+\underline{y},\overline{x}+r]\) and its density  \(h\) is well defined and strictly positive over this interval. Let \(D(p,p^*,n)\) be the demand for a firm when other firms set \(p^*\) and it sets \(p\). According to the analysis before, it can be calculated as follows:
\begin{align*}
     D(p,p^*,n)=[1-H(\sigma-p^*+p)]\frac{L}{n}\sum_{j=0}^{n-1}H(\sigma)^j+L\int_{p}^{\sigma-p^*+p}H(w-p+p^*)^{n-1}h(w)dw,
\end{align*}
where the two terms correspond to the two parts of demand respectively. First order condition gives a necessary condition for the optimal \(p\):
\begin{align*}
  \frac{\partial \pi(p,p^*,n)}{\partial p}&=D_p(p,p^*,n)[p-(-\frac{D(p,p^*,n)}{D_p(p,p^*,n)})]\\
    &=D_p(p,p^*,n)\left[p-\frac{\frac{L}{n}[1-H(\sigma-p^*+p)]\frac{1-H(\sigma)^n}{1-H(\sigma)}+L\int_{p}^{\sigma-p^*+p}H(w-p+p^*)^{n-1}h(w)dw}{\frac{L}{n}h(\sigma-p^*+p)\frac{1-H(\sigma)^n}{1-H(\sigma)}+L\int_{p}^{\sigma-p^*+p}H(w-p+p^*)^{n-1}h'(w)dw}\right]\\
    &=0.   
\end{align*}

And the firm has no incentive to deviate only if \(p=p^*\) satisfies this condition. I focus on monopolistic competition, i.e., the case when \(n=\infty\).\footnote{Previously, I consider  trees with a finite number of terminal nodes. Even there is an infinite number of terminal nodes under monopolistic competition, all the results in Section \ref{s3} go through.} In this case, the probability of the second event is 0 as a buyer never samples all firms and then returns to a previously visited firm. Then by setting \(p=p^*\), a candidate  equilibrium price is obtained, which is \(p^*=\frac{1-H(\sigma)}{h(\sigma)}\). As FOC is only  necessary for optimality, additional assumptions on \(H\) are required to ensure sufficiency. I impose the following assumption:

 \begin{assumption}\label{logconcave}
     \(f\) and \(g\) are log-concave.
 \end{assumption}
This is a technical assumption imposed on the two primitive distributions to ensure a desirable property of \(H\), the distribution of interest, as stated in the following lemma:

\begin{lemma}\label{inchazard}
Under Assumption \ref{logconcave}, \(H\) has increasing hazard rate.    
\end{lemma}


Also, the search cost has to be small enough so that buyers are willing to search in equilibrium. Specifically, \(\sigma\) has to be larger than the threshold \(\sigma^*\) defined in the following lemma:
\begin{lemma}\label{smallcost}
Under Assumption \ref{logconcave}, 
\begin{itemize}
\item either \(w \geq \frac{1-H(w)}{h(w)}\), \(\forall w \in [\underline{x}+\underline{y},\overline{x}+\overline{y}]\); in this case, let \(\sigma^*=\underline{x}+\underline{y}\);
\item or there exists an \(\sigma^*>0\) such that \(\frac{1-H(\sigma^*)}{h(\sigma^*)}=\sigma^*\).
\end{itemize}
\end{lemma}

Now it is ready to state the existence of the equilibrium:

\begin{proposition}\label{equi1}
Under Assumption \ref{logconcave}, if the search costs \(c_X\) and \(c_Y\) are such that \(\sigma>\sigma^*\), then under monopolistic competition, there exists a unique symmetric pure strategy  equilibrium with active consumer search, which is \(p^*=\frac{1-H(\sigma)}{h(\sigma)}\). 

\end{proposition}


\subsection{Price revealed after stage 2}\label{s43}
In this subsection, I study the case when the price is revealed together with \(Y\) at stage 2. Suppose all but one firm sets \(p^{\dagger}\) and the remaining firm, say firm \(i\), sets \(p\). Then 
\begin{align*}
    \kappa_i&=\min \{\sigma-p^{\dagger},X_i+r-p^{\dagger},X_i+Y_i-p\},\\
    \kappa_j&=\min \{\sigma-p^{\dagger},X_i+r-p^{\dagger},X_i+Y_i-p^{\dagger}\} \text{ for }j \neq i.
\end{align*}
Similar to the analysis before, the demand for firm \(i\) is
\begin{align*}
     D(p,p^{\dagger},n)=\Phi(\sigma,p) \frac{1}{n} \sum_{j=0}^{n-1}H(\sigma)^{j}+\int_{p^{\dagger}}^{\sigma}[\Phi(w,p)-\Phi(\sigma,p)](n-1)H(w)^{n-2}h(w)dw,
\end{align*}
where 
\begin{align*}
\Phi(w,p)&=\mathrm{Pr}(X+r \geq w \text{ and }X+Y \geq w-p^{\dagger}+p)\\
&=\int_{w-r}^{\infty} \int_{w-p^{\dagger}+p-x}^{\infty}dG(y) dF(x)\\
&=\int_{w-r}^{\infty} [1-G(w-p^{\dagger}+p-x)]dF(x).
\end{align*}

When the price is only revealed at stage 2, the demand responds less aggressively to a change in the price, and FOC gives a necessary condition for the optimal \(p\):

\begin{align*}
\frac{\partial \pi(p,p^{\dagger},n)}{\partial p}&=D_p(p,p^{\dagger},n)[p-(-\frac{D(p,p^{\dagger},n)}{D_p(p,p^{\dagger},n)})]\\
&=D_p(p,p^{\dagger},n)\left[p-\frac{\frac{1}{n}\Phi(\sigma,p)\frac{1-H(\sigma)^n}{1-H(\sigma)}+\int_{p^{\dagger}}^{\sigma}[\Phi(w,p)-\Phi(\sigma,p)](n-1)H(w)^{n-2}h(w)dw}{-\frac{1}{n}\frac{\partial \Phi(\sigma,p)}{\partial p}\frac{1-H(\sigma)^n}{1-H(\sigma)}+\int_{p^{\dagger}}^{\sigma}(-\frac{\partial \Phi(w,p)}{\partial p}+\frac{\partial \Phi(\sigma,p)}{\partial p})(n-1)H(w)^{n-2}h(w)dw}\right]\\
&=0.
\end{align*}
When \(n = \infty\), the \(p=p^{\dagger}\) that satisfies this condition is
\begin{align*}
p^{\dagger}&=-\frac{\Phi(\sigma,p^{\dagger})}{\frac{\partial \Phi(\sigma,p)}{\partial p}|_{p=p^{\dagger}}}\\
    &=\frac{\int_{\sigma-r}^{\infty}[1-G(\sigma-x)]dF(x)}{\int_{\sigma-r}^{\infty}g(\sigma-x)dF(x)}.
\end{align*}

\begin{proposition}\label{equi2}
    Under Assumption \ref{logconcave}, if the search costs \(c_X\) and \(c_Y\) are such that \(\sigma>\frac{1-G(r)}{g(r)}\), then under monopolistic competition, there exists a unique symmetric pure strategy equilibrium with active consumer search, which is \(p^{\dagger}=\frac{\int_{\sigma-r}^{\infty}[1-G(\sigma-x)]dF(x)}{\int_{\sigma-r}^{\infty}g(\sigma-x)dF(x)}\).
\end{proposition}

Moreover, the equilibrium price is higher than when price is revealed in stage 1, because an increase in price is not observed immediately, and consumers respond less aggressively since they have already incurred search costs in both stages.

\begin{proposition}\label{stage2higher}
    \(p^* \leq p^{\dagger}\).
\end{proposition}

\subsection{Comparison with the benchmark}\label{s44}
In this subsection, I compare the equilibrium price and consumer surplus under monopolistic competition to the benchmark case without gradual resolution of uncertainty. To make a fair comparison, I assume the search cost in the benchmark is \(c_X+c_Y\). Similar analysis gives the equilibrium price in the benchmark, which I denote using \(\hat{p}\). Let \(K\) be the distribution of \(Z\) and \(k\) its density, \(\hat{p}=\frac{1-K(\hat{\sigma})}{k(\hat{\sigma})}\), where \(\mathbb{E}[(X+Y-\hat{\sigma})^+]=c_X+c_Y\).

\subsubsection{\(p^*\) v.s. \(\hat{p}\)}
I identify two market forces that lead to \(p^* \leq \hat{p}\), each detailed in Lemma \ref{option} and Lemma \ref{elasticity}. First, recall that \(\sigma\) defined for the first stage of inspection roughly captures the potential benefit of sampling another product. With gradual resolution of uncertainty, a buyer can always stop after seeing an undesirable realization of \(X\) and save on \(c_Y\). This option value makes it more worthwhile for the buyer to sample new products, which is formally stated as follows:

\begin{lemma}\label{option}
    \(\sigma \geq \hat{\sigma}\).
\end{lemma}

Second, when the price is revealed with \(X\), while \(Y\) is only revealed at the second stage, demand responds more aggressively than the benchmark. This is because after completing the first stage of inspection for some firm \(i\), the buyer's decision on whether to continue with firm \(i\) only depends on \(X\) and the price. If this price is too high that the buyer would rather sample another firm, he will never return to firm \(i\). This is different than the benchmark case in which everything is revealed at the same time and buyers with a high enough realization of \(Y\) will buy from firm \(i\) even given its high price. Simple calculation shows that the elasticity at price \(p\)  are \(\epsilon=\frac{ph(p)}{1-H(p)}\) and \(\hat{\epsilon}=\frac{pk(p)}{1-K(p)}\), which is captured  by the hazard rates of the distributions. Therefore, the intuition on demand elasticity can be equivalently stated as a comparison of  hazard rates:

\begin{lemma}\label{elasticity}
    \(H\) has a higher hazard rate than \(K\).
\end{lemma}

The two lemmas combine to give the following theorem, implying that gradual resolution of uncertainty helps reduce the market price.
\begin{theorem}\label{lowerprice}
    \(p^* \leq \hat{p}\).
\end{theorem}

Moreover, because consumers also benefit from the option value, and the equilibrium price is lower, consumer surplus is higher than in the benchmark.

\begin{corollary}\label{highercs}
When the price is revealed at the first stage of inspection, gradual resolution of uncertainty gives a higher consumer surplus.
\end{corollary}

\begin{remark}
    Note that the above analysis holds true only for \(n=\infty\), i.e., the case of monopolistic competition. When \(n\) is small, even if the buyer is temporarily turned away by a high price or a bad realization of the match value at firm \(i\), it is still possible for him to return and purchase from firm \(i\) when he gets a bad match value at every other firm in the market. When \(n=\infty\), this event happens with probability close to \(0\). I discuss how Theorem \ref{lowerprice} fails in small markets in Section \ref{s51}.
\end{remark}

\subsubsection{\(p^{\dagger}\) v.s. \(\hat{p}\)}
Notice that,
\begin{align*}
    \hat{p}=\frac{1-K(\hat{\sigma})}{k(\hat{\sigma})}=\frac{\int_{-\infty}^{\infty}[1-G(\hat{\sigma}-x)]dF(x)}{\int_{-\infty}^{\infty}g(\hat{\sigma}-x)dF(x)},
\end{align*}
and even under Assumption \ref{logconcave}, comparisons of prices and consumer surplus are ambiguous. Intuitively, there are two forces that affect equilibrium price. The first force is similar to the previous case: gradual resolution of uncertainty intensifies competition as the option value makes it more beneficial to sample another product. The second force is, however, dramatically different from before. Since in this case price is revealed only at the second stage, raising the price will never deter consumer search. Moreover, given that the buyer chooses to complete the second stage of inspection, the firm knows that the realization of \(X\) must be high enough so that he chooses to continue rather than sample another firm. Thus, the firm has an incentive to raise its price. The two forces go in opposite directions, and whether equilibrium price goes up or down depends on the search environment. Consumer surplus can fall when the increase in price brought about by locking the consumer in from stage 1 fully offsets the option value.

\subsection{When to reveal the price?}\label{s45}
In the preceding analysis, I fix the same price revelation policy across firms and examine the resulting market equilibrium price. I also implicitly assume that firms can commit to the price set in the initial stage. Two natural challenges arise from this setup. First, firms typically have significant control over their price revelation strategies, which should be determined endogenously as part of the equilibrium. Second, firms may lack commitment power and choose to set a different price in the second stage.

To address the first challenge, I consider a modified pricing game in which firms choose not only the price to set, but also the timing of its revelation. Specifically, each firm \(i\) can either reveal its price \(p_{2,i}\) only at stage 2, or reveal an initial price \(p_{1,i}\) at stage 1 followed by a potentially different price \(p_{2,i}\) at stage 2. The price revelation policy is unknown to the buyer at the outset and it is discovered along search. To address the second challenge, I analyze two scenarios separately: (i) firms cannot commit to the price set in the initial stage; and (ii) firms are prohibited from raising their price at stage 2. Case (ii) reflects a regulatory environment that bans drip pricing, while case (i) corresponds to a market without such regulations. Still, I focus on monopolistic competition.

Before delving into the detail,  I outline the equilibrium concept I will use for the analysis. An equilibrium consists of firms' strategies, the buyer's belief and his search policy, satisfying the following conditions:
\begin{enumerate}
    \item At every history, the buyer adopts a search policy that maximizes his expected payoff, given his belief.
    \item For each firm \(i\), the strategy chosen maximizes its expected profit, given the buyer’s search behavior.
    \item The buyer’s belief is updated according to Bayes’ rule whenever applicable. At off-path histories (when a deviation is detected at stage 1), belief is updated according to the following assumptions:
\end{enumerate}
\begin{assumption}\label{beliefunreg}
    In an unregulated market, upon any off-path deviation at stage 1, the consumer believes the price charged at stage 2 is \(\infty\). 
\end{assumption}
This assumption says that consumers become extremely skeptical after seeing a deviation at stage 1 and will believe the price at stage 2 is as high as possible. It is a very restrictive refinement and not necessary for any of the results developed in this subsection. In Section \ref{s52}, I show that all results remain true under much weaker restrictions on belief refinement.
\begin{assumption}\label{beliefreg}
    In a regulated market, upon any off-path deviation at stage 1, the consumer believes the price charged at stage 2 is exactly the price charged at stage 1.
\end{assumption}
In a regulated market, consumers still become pessimistic about the price when observing a deviation, but due to the regulation, they know there is an upperbound on what can be charged at stage 2 and will believe that price is exactly the same as the stage 1 price.

Equipped with the equilibrium concept, I characterize the equilibrium price in both unregulated and regulated markets. I show that the unique equilibrium price in the unregulated (regulated) market coincides with the one derived in Section \ref{s43} (\ref{s42}), respectively.

\subsubsection{Unregulated markets}
Suppose that firm \(i\) playing mixed strategy \(\rho_i \in \Delta((\varnothing \cup \mathbb{R}_{+}) \times \mathbb{R}_{+})\) constitutes an equilibrium, where \(\varnothing\) means the firm chooses not to reveal any price at stage 1. Intuitively, since firms have no commitment power, prices revealed at stage 1 are only cheap talk messages, and the buyer might as well just ignore them. Thus, a natural guess of equilibrium in such a market is the one obtained in Section \ref{s43}, where no price is revealed at stage 1. However, this intuition is incomplete. Given a putative equilibrium, price revealed at stage 1 can be informative of the price to be revealed at stage 2 and thus coordinate consumer belief, whereas in Section \ref{s43}, there is no such thing that can change the buyer's belief before stage 2. Thus one cannot confidently rule out mixed strategy equilibria.

To pin down the equilibrium, I first show that in any equilibrium, regardless of the price revealed at stage 1 (including \(\varnothing\)), the price set at stage 2 is unique, as stated in the following lemma:

\begin{lemma}\label{nomix2}
In any equilibrium, for any price revealed at stage 1, the price set at stage 2 must be unique.
\end{lemma}

To understand this lemma, note that given a putative equilibrium, for any  price revealed at stage 1, the buyer will form a belief about the price to be revealed later, and the firm will always set a price at stage 2 to maximize its profit given this belief. Assumption \ref{logconcave} ensures that the distribution is well-behaved so that there is a unique maximizer. 

Lemma \ref{nomix2} states only that, conditional on the price revealed at stage 1, the firm does not mix over different prices at stage 2. However, it does not rule out mixed strategies altogether; in particular, the firm may still randomize over different pairs of prices. The following lemma establishes that such randomization over prices at stage 2 cannot be an equilibrium phenomenon:

\begin{lemma}\label{uniquestage2}
    In any equilibrium, a unique price is revealed at stage 2.
\end{lemma}

In order for some price \(p\) (revealed at stage 2) to be part of an equilibrium, it must be that, given the buyer's belief that the firm will charge \(p\), it is still optimal for the firm to do so. Under Assumption \ref{logconcave}, by leveraging the variation diminishing property in \cite{karlin1968total}, I show that,  given any index for the first stage of inspection, there is a unique \(p\) such that the firm has no incentive to deviate at stage 2, thus the firm will not randomize over prices at stage 2. 

Based on Lemma \ref{nomix2} and \ref{uniquestage2}, the search for equilibria reduces to the analysis in Section \ref{s43}, and one can easily verify that it is indeed an equilibrium price in an unregulated market with the belief updating rule. Thus, one has the following proposition:

\begin{proposition}\label{unique2}
Under Assumption \ref{logconcave} and \ref{beliefunreg}, in an unregulated market, the unique equilibrium price under monopolistic competition is exactly \(p^{\dagger}\).    
\end{proposition}

\subsubsection{Regulated markets}
I follow similar procedures in analyzing the equilibrium in regulated markets. Similar to Lemma \ref{nomix2}, Assumption \ref{logconcave} ensures the profit function is single-peaked in stage 2 price, and thus there is a unique maximizer under the regulation, as stated in the following lemma:

\begin{lemma}\label{nomix1}
    In any equilibrium, for any price revealed at stage 1, the price set at stage 2 must be unique.
\end{lemma}

Based on Lemma \ref{nomix1}, conditional on the price revealed at stage 1, the firm does not mix over different prices at stage 2. But this does not rule out mixed strategy equilibrium altogether: it is possible for the firm to mix over multiple pairs of prices. In Lemma \ref{uniquestage}, I show that this is not an equilibrium phenomenon. Intuitively, if the firm mixed over different price pairs, they must give the same profit. Since the profit function is single-peaked, there can be at most 2 price pairs, and they must fall on opposite sides of the peak. But the firm can profitably deviate by setting the price at both stages at the peak.

\begin{lemma}\label{uniquestage}
    In any equilibrium, a unique price is revealed at stage 2.
\end{lemma}

Based on Lemma \ref{nomix1} and \ref{uniquestage}, the search for equilibria reduces to the analysis in Section \ref{s42}, and one can easily verify that it is indeed an equilibrium price in a regulated market with the belief updating rule. Thus, one has the following proposition:

\begin{proposition}\label{unique1}
    Under Assumption \ref{logconcave} and \ref{beliefreg}, in a regulated market, the unique equilibrium price under monopolistic competition is exactly \(p^*\).
\end{proposition}

In both markets, buyers benefit from the option value; however, regulation leads to a lower market price, thereby increasing consumer surplus, as stated below:

\begin{corollary}\label{effregulation}
Under Assumption \ref{logconcave}, \ref{beliefunreg} and \ref{beliefreg}, the regulation on hidden fees increases consumer surplus.    
\end{corollary}


\section{Discussion}\label{s5}


\subsection{Small markets}\label{s51}
As noted earlier, Theorem \ref{lowerprice} relies on monopolistic competition (i.e., \(n=\infty\)). Below, I provide a duopoly example in which the option to return to a previously visited firm overturns the conclusion of Theorem \ref{lowerprice}.

Let \(n=2\), \(L=1\) and \(F=G=U[0,1]\). Using the FOC in Section \ref{s42}, a candidate equilibrium price with gradual resolution of uncertainty is determined by the following equation:
\begin{align*}
    p^*=\frac{\frac{1}{2}[1-H(\sigma)^2]+\int_{p^*}^{\sigma}H(w)h(w)dw}{\frac{1}{2}h(\sigma)[1+H(\sigma)]+\int_{p^*}^{\sigma}H(w)h'(w)dw}.
\end{align*}
And a candidate equilibrium price in the benchmark satisfies
\begin{align*}
    \hat{p}=\frac{\frac{1}{2}[1-K(\hat{\sigma})^2]+\int_{\hat{p}}^{\hat{\sigma}}K(z)k(z)dz}{\frac{1}{2}k(\hat{\sigma})[1+K(\hat{\sigma})]+\int_{\hat{p}}^{\hat{\sigma}}K(z)k'(z)dz}.
\end{align*}
It can be calculated that
\begin{align*}
    H(w)=
\begin{cases}
0, & w<0,\\
\dfrac{w^2}{2}, & 0\le w\le r,\\
w-r+\dfrac{r^2}{2}, & r< w \le 1,\\
-\frac{1}{2}w^2+2w+\frac{1}{2}r^2-r-\frac{1}{2}, & 1< w \le 1+r,\\
1, & w>1+r.
\end{cases}
&
\quad
h(w) =
\begin{cases}
0, & w < 0, \\
w, & 0 \le w \le r, \\
1, & r < w \le 1, \\
2 - w, & 1 < w \le 1+r, \\
0, & w > 1+r,
\end{cases}
\end{align*}
and
\begin{align*}
    K(z)=
\begin{cases}
0, & z < 0, \\
\dfrac{z^2}{2}, & 0 \le z \le 1, \\
1 - \dfrac{(2-z)^2}{2}, & 1 < z \le 2, \\
1, & z > 2.
\end{cases}
&
\quad
k(z) =
\begin{cases}
z, & 0 \le z \le 1, \\
2-z, & 1 < z \le 2, \\
0, & \text{otherwise}.
\end{cases}
\end{align*}

Let \(c_X=0.05\) and \(c_Y=0.1\), we get \(r=1-\frac{1}{\sqrt{5}}\approx 0.5528\), \(\sigma \approx 1.1393\) and \(\hat{\sigma}\approx 1.0345\). Plugging these values back in the FOCs and solve for the fixed points, I get \(p^*\approx 0.6989\) and \(\hat{p}\approx 0.5671\). Clearly, \(p^* > \hat{p}\), which contradicts the result in Theorem \ref{lowerprice}.

\subsection{Equilibrium refinement}\label{s52}

In Section \ref{s45}, I assume that, in an unregulated market, upon a deviation at stage 1, the consumer assigns probability one to an infinite final price. While this assumption simplifies the analysis of first stage deviation, it is admittedly extreme. In particular, it may allow some equilibria to survive only because deviations are punished excessively under this belief. To address this concern, I replace it with a weaker assumption that imposes a more plausible restriction on the consumer’s beliefs following a deviation.
 It is stated as follows:
\begin{assumption}\label{weaker}
Given a putative equilibrium with a first stage index \(\sigma-\tilde{p}\), upon a deviation at stage 1, the consumer believes the price charged at stage 2 is weakly higher than \(\tilde{p}\).
\end{assumption}
This assumption says that, upon a deviation at stage 1, the consumer becomes skeptical and expects the final price to be weakly higher than before. It nests Assumption \ref{beliefunreg} as a special case. I illustrate the robustness of previous results by showing that they remain true under this weaker assumption:

\begin{proposition}\label{unique2weaker}
Under Assumption \ref{logconcave} and \ref{weaker}, in an unregulated market, the unique equilibrium price under monopolistic competition is exactly \(p^{\dagger}\).
\end{proposition}

This proposition shows that the equilibrium characterization in unregulated markets, and thus the effect of drip pricing regulation does not depend on the extreme off-path belief. As long as a stage 1 deviation makes the consumer weakly more pessimistic about the final price, the equilibrium price and welfare conclusions remain unchanged. In this sense, the earlier results are robust to a broad class of reasonable belief restrictions. The argument larges follows earlier analysis. By ruling out profitable stage 2 deviations, there is a unique candidate equilibrium in this case. It then suffices to check that no profitable deviation exists at stage 1 from this candidate. Under the weaker belief restriction, a stage 1 deviation makes the buyer weakly more pessimistic about the final price, which shifts the deviating firm’s demand curve downward. As a result, even when the firm best responds to the buyer’s belief, its deviation profit is lower.

\begin{remark}
    Note that Assumption \ref{beliefreg} cannot be relaxed in the same way; requiring only a weakly more pessimistic belief is too weak to deter profitable deviation at stage 1. The reason is that charging \(p^*\) is a best response only under the regulation; absent the price restriction, the firm strictly prefers a higher price. Thus, by deviating to a higher stage 1 price, the firm can make a higher profit as long as the consumer does not become too skeptical, because the deviation  allows the firm to charge a higher final price.
\end{remark}

\section{Conclusion}\label{s6}
This paper studies a search problem represented as a tree that features both gradual resolution of uncertainty and hierarchical similarities among prizes. The problem admits significant tractability in that an index solution achieves the optimum. It is also practically relevant, capturing decision problems that arise in settings such as R\&D and e-commerce.

This paper opens several avenues for future investigation, both theoretical and empirical. One could examine the optimal policy under looser assumptions, such as nonobligatory inspection, relaxed ordering constraints, or random tree structures. The framework may enable the study of environments that are difficult to accommodate with traditional search models. I leave these extensions to future research.

\newpage
\bibliography{references}

\newpage
\appendix
\setcounter{secnumdepth}{0}
\setcounter{equation}{0}
\renewcommand{\theequation}{A\arabic{equation}}
\section{Appendix}\label{sa}
\noindent\textbf{Proof of Lemma \ref{upperbound}}: I first state a useful lemma and present its proof here. Similar result is proved in Lemma 7 in \cite{bowers2024matching}.
\begin{lemma}\label{condind}
For any \(i_1,\cdots,i_k\),  \(\mathbb{E}\Biggl[\mathbb{I}_{i_1,\cdots,i_{k}}c_{i_1,\cdots,i_k}\Biggr]=\mathbb{E}\left[\mathbb{I}_{i_1,\cdots,i_k}\Biggl(\max_{1 \leq l \leq N_{i_1,\cdots,i_k}} \kappa_{i_1,\cdots,i_k,l}-\sigma_{i_1,\cdots,i_k}\Biggr)^{+}\right]\).   
\end{lemma}
\noindent\textbf{Proof of Lemma \ref{condind}}:
\begin{align*}   &\mathbb{E}\Biggl[\mathbb{I}_{i_1,\cdots,i_{k}}c_{i_1,\cdots,i_k}\Biggr]\\
    =&\mathbb{E}\Biggl[\mathbb{E}\left[\mathbb{I}_{i_1,\cdots,i_{k}}\Biggl|X_{i_1},\cdots,X_{i_1,\cdots,i_{k-1}}\right]c_{i_1,\cdots,i_{k}}\Biggr]\\
    =&\mathbb{E}\Biggl[\mathbb{E}\left[\mathbb{I}_{i_1,\cdots,i_k}\Biggl|X_{i_1},\cdots,X_{i_1,\cdots,i_{k-1}}\right]\mathbb{E}\left[\Biggl(\max_{1 \leq l \leq N_{i_1,\cdots,i_k}} \kappa_{i_1,\cdots,i_k,l}-\sigma_{i_1,\cdots,i_k}\Biggr)^{+}\Biggl|X_{i_1},\cdots,X_{i_1,\cdots,i_{k-1}}\right]\Biggr]\\
    =&\mathbb{E}\Biggl[\mathbb{E}\left[\mathbb{I}_{i_1,\cdots,i_k}\Biggl(\max_{1 \leq l \leq N_{i_1,\cdots,i_k}} \kappa_{i_1,\cdots,i_k,l}-\sigma_{i_1,\cdots,i_k}\Biggr)^{+}\Biggl|X_{i_1},\cdots,X_{i_1,\cdots,i_{k-1}}\right]\Biggr]\\
    =&\mathbb{E}\left[\mathbb{I}_{i_1,\cdots,i_k}\Biggl(\max_{1 \leq l \leq N_{i_1,\cdots,i_k}} \kappa_{i_1,\cdots,i_k,l}-\sigma_{i_1,\cdots,i_k}\Biggr)^{+}\right].
\end{align*}
The first and last equality are due to the law of total expectation. The second equality comes from the definition of \(\kappa\) and \(\sigma\). The third equality is due to Assumption \ref{independence}. \qed

\bigskip

The proof of the lemma follows from induction. Formally, I prove that the agent's expected payoff from the subtree rooted at any node \(S_{i_1,\cdots,i_k} \notin \mathcal{T}\cup\{R\}\) is bounded by
\begin{align*}
    \mathbb{E}\left[\left(\sum_{S_{j_1,\cdots,j_n}\in \mathcal{T} \cap \mathcal{D}_{i_1,\cdots,i_k}}\mathbb{A}_{j_1,\cdots,j_n}\right)\kappa_{i_1,\cdots,i_k}\right].
\end{align*}

\noindent\textbf{Step 1:} Consider any \(S_{i_1,\cdots,i_n} \in \mathcal{T}\), then the expected payoff from this subtree can be bounded as follows:
\begin{align*}
    &\mathbb{E}\Biggl[\mathbb{A}_{i_1,\cdots,i_{n}}v^{i_1,\cdots,i_{n}}-\mathbb{I}_{i_1,\cdots,i_{n}}c_{i_1,\cdots,i_{n}}\Biggr]\\
    =&\mathbb{E}\Biggl[\mathbb{A}_{i_1,\cdots,i_{n}}v^{i_1,\cdots,i_{n}}-\mathbb{I}_{i_1,\cdots,i_{n}}\Biggl(\kappa_{i_1,\cdots,i_{n},1}-\sigma_{i_1,\cdots,i_{n}}\Biggr)^{+}\Biggr]\\
    \leq&\mathbb{E}\Biggl[\mathbb{A}_{i_1,\cdots,i_{n}}v^{i_1,\cdots,i_{n}}-\mathbb{A}_{i_1,\cdots,i_{n}}\Biggl(\kappa_{i_1,\cdots,i_{n},1}-\sigma_{i_1,\cdots,i_{n}}\Biggr)^{+}\Biggr]\\
    =&\mathbb{E}\Biggl[\mathbb{A}_{i_1,\cdots,i_{n}}\min \Biggl\{\kappa_{i_1,\cdots,i_{n},1},\sigma_{i_1,\cdots,i_{n}}\Biggr\}\Biggr]\\
    =&\mathbb{E}\Biggl[\mathbb{A}_{i_1,\cdots,i_{n}}\kappa_{i_1,\cdots,i_{n}}\Biggr],
\end{align*}
where the first equality comes from Lemma \ref{condind}, the inequality comes from \(\mathbb{A}_{i_1,\cdots,i_n} \leq \mathbb{I}_{i_1,\cdots,i_n}\) since inspection is obligatory. The remaining comes from definition of \(\kappa\) and \(\sigma\).

\noindent\textbf{Step 2:} Suppose the bound is valid for the subtree rooted at \(S_{i_1,\cdots,i_k,l}\) with \(l=1,\cdots,N_{i_1,\cdots,i_k}\). I show that the bound is valid for \(S_{i_1,\cdots,i_k}\). Note that, the expected payoff from the subtree rooted at \(S_{i_1,\cdots,i_k}\) is the sum of the expected payoff from the subtrees rooted at \(S_{i_1,\cdots,i_k,l}\) minus the search cost spent on \(e_{i_1,\cdots,i_k}\). Therefore, the expected payoff can be bounded as follows:
\begin{align*}
    &\mathbb{E}\Biggl[\sum_{l=1}^{N_{i_1,\cdots,i_k}}\left(\sum_{S_{j_1,\cdots,j_n} \in \mathcal{T}\cap\mathcal{D}_{i_1,\cdots,i_k,l}}\mathbb{A}_{j_1,\cdots,j_n}\right)\kappa_{i_1,\cdots,i_k,l} \Biggr]-\mathbb{E}\Biggl[\mathbb{I}_{i_1,\cdots,i_k}c_{i_1,\cdots,i_k}\Biggr]\\
    =&\mathbb{E}\Biggl[\sum_{l=1}^{N_{i_1,\cdots,i_k}}\left(\sum_{S_{j_1,\cdots,j_n} \in \mathcal{T}\cap\mathcal{D}_{i_1,\cdots,i_k,l}}\mathbb{A}_{j_1,\cdots,j_n}\right)\kappa_{i_1,\cdots,i_k,l}-\mathbb{I}_{i_1,\cdots,i_k}\left(\max_{1 \leq l \leq N_{i_1,\cdots,i_k}} \kappa_{i_1,\cdots,i_k,l}-\sigma_{i_1,\cdots,i_k}\right)^{+}\Biggr]\\
    \leq&\mathbb{E}\Biggl[\left(\sum_{S_{j_1,\cdots,j_n} \in \mathcal{T} \cap \mathcal{D}_{i_1,\cdots,i_k}} \mathbb{A}_{j_1,\cdots,j_n}\right)\max_{1 \leq l \leq N_{i_1,\cdots,i_k}}\kappa_{i_1,\cdots,i_k,l} -\mathbb{I}_{i_1,\cdots,i_k}\left(\max_{1 \leq l \leq N_{i_1,\cdots,i_k}} \kappa_{i_1,\cdots,i_k,l}-\sigma_{i_1,\cdots,i_k}\right)^{+}\Biggr]\\
    \leq&\mathbb{E}\Biggl[\left(\sum_{S_{j_1,\cdots,j_n} \in \mathcal{T} \cap \mathcal{D}_{i_1,\cdots,i_k}} \mathbb{A}_{j_1,\cdots,j_n}\right)\left[\max_{1 \leq l \leq N_{i_1,\cdots,i_k}}\kappa_{i_1,\cdots,i_k,l}-\left(\max_{1 \leq l \leq N_{i_1,\cdots,i_k}} \kappa_{i_1,\cdots,i_k,l}-\sigma_{i_1,\cdots,i_k}\right)^{+}\right]\Biggr]\\
    =&\mathbb{E}\Biggl[\left(\sum_{S_{j_1,\cdots,j_n} \in \mathcal{T} \cap \mathcal{D}_{i_1,\cdots,i_k}} \mathbb{A}_{j_1,\cdots,j_n}\right) \min \Biggl\{\max_{1 \leq l \leq N_{i_1,\cdots,i_k}}\kappa_{i_1,\cdots,i_k,l},\sigma_{i_1,\cdots,i_k}\Biggr\}\Biggr]\\
    =&\mathbb{E}\Biggl[\left(\sum_{S_{j_1,\cdots,j_n} \in \mathcal{T} \cap \mathcal{D}_{i_1,\cdots,i_k}} \mathbb{A}_{j_1,\cdots,j_n}\right) \kappa_{i_1,\cdots,i_k}\Biggr],
\end{align*}
where the second inequality comes from obligatory inspection and unit demand.

\noindent\textbf{Step 3: }By induction, the bound is also valid for any subtree rooted at the child node of the root. The expected payoff for the entire tree is the sum of the expected payoff from these trees, so the bound on the total expected payoff is \(\mathbb{E}\left[\sum_{l=1}^{N_0}\left(\sum_{S_{i_1,\cdots,i_n} \in \mathcal{T} \cap \mathcal{D}_l}\mathbb{A}_{i_1,\cdots,i_n}\right)\kappa_{l}\right]\). This finishes the proof.\qed

\bigskip

\noindent\textbf{Proof of Lemma \ref{nonexposure}}: In order to prove the bound is tight, I show all the inequalities in the proof of Lemma \ref{upperbound} are indeed equalities. 

\noindent\textbf{Part 1: }For \(S_{i_1,\cdots,i_n} \in \mathcal{T}\), under the index policy, if \(\mathbb{I}_{i_1,\cdots,i_n}=1\) but \(\mathbb{A}_{i_1,\cdots,i_n}=0\), then it must be that \(\kappa_{i_1,\cdots,i_n,1} \leq \sigma_{i_1,\cdots,i_n}\). Suppose to the contrary \(\kappa_{i_1,\cdots,i_n,1} > \sigma_{i_1,\cdots,i_n}\). Since \(\mathbb{I}_{i_1,\cdots,i_n}=1\), it means \(\sigma_{i_1,\cdots,i_n}\) is weakly higher than any other indexes. As \(v^{i_1,\cdots,i_n}(x_{i_1},\cdots,x_{i_1,\cdots,i_n})=\kappa_{i_1,\cdots,i_n,1} > \sigma_{i_1,\cdots,i_n}\), the final prize of \(S_{i_1,\cdots,i_n}\) is strictly higher than any other indexes, then according to the index policy, the agent should stop search and take the prize. This contradicts \(\mathbb{A}_{i_1,\cdots,i_n}=0\). For \(S_{i_1,\cdots,i_k} \notin \mathcal{T}\), under the index policy, if \(\mathbb{I}_{i_1,\cdots,i_k}=1\) but \(\sum_{S_{j_1,\cdots,j_n} \in \mathcal{T} \cap \mathcal{D}_{i_1,\cdots,i_k}}\mathbb{A}_{j_1,\cdots,j_n}=0\), then it must be \(\max_{1\leq l \leq N_{i_1,\cdots,i_k}} \kappa_{i_1,\cdots,i_k,l} \leq \sigma_{i_1,\cdots,i_k}\). Suppose to the contrary \(\max_{1\leq l \leq N_{i_1,\cdots,i_k}} \kappa_{i_1,\cdots,i_k,l} > \sigma_{i_1,\cdots,i_k}\). By definition, \(\kappa_{i-1,\cdots,i_k,l} \geq \sigma_{i_1,\cdots,i_k,l}\), therefore under the index policy the agent will inspect at least one of its child nodes. By induction, there exists at least one terminal node such that the indexes of all the edges connecting it to \(S_{i_1,\cdots,i_k}\) and the final prize are  strictly greater than \(\sigma_{i_1,\cdots,i_k}\). Under the index policy, the agent will claim a prize from \(\mathcal{T}\cap\mathcal{D}_{i_1,\cdots,i_k}\), which contradicts \(\sum_{S_{j_1,\cdots,j_n} \in \mathcal{T} \cap \mathcal{D}_{i_1,\cdots,i_k}}\mathbb{A}_{j_1,\cdots,j_n}=0\).

\noindent\textbf{Part 2: }I show that under the index policy, the following holds:
\begin{align*}
    &\mathbb{E}\Biggl[\sum_{l=1}^{N_{i_1,\cdots,i_k}}\left(\sum_{S_{j_1,\cdots,j_n} \in \mathcal{T}\cap\mathcal{D}_{i_1,\cdots,i_k,l}}\mathbb{A}_{j_1,\cdots,j_n}\right)\kappa_{i_1,\cdots,i_k,l}\Biggr]\\
    =&\mathbb{E}\Biggl[\left(\sum_{l=1}^{N_{i_1,\cdots,i_k}}\sum_{S_{j_1,\cdots,j_n} \in \mathcal{T}\cap\mathcal{D}_{i_1,\cdots,i_k,l}}\mathbb{A}_{j_1,\cdots,j_n}\right) \max_{1 \leq l \leq N_{i_1,\cdots,i_k}} \kappa_{i_1,\cdots,i_k,l}\Biggr].
\end{align*}
To proceed, I first state and prove the following lemma:
\begin{lemma}\label{maxminpath}
    The following relation holds true:
    \begin{align*}
        \min \Biggl\{\sigma_{i_1},\sigma_{i_1,i_2},\cdots,\sigma_{i_1,\cdots,i_{k-1}},\kappa_{i_1,\cdots,i_k}\Biggr\}=\max_{S_{j_1,\cdots,j_n} \in \mathcal{T}\cap\mathcal{D}_{i_1,\cdots,i_k}} \min \left\{\sigma_{j_1},\sigma_{j_2},\cdots,\sigma_{j_1,\cdots,j_n},\sigma_{j_1,\cdots,j_n,1}\right\}.
    \end{align*}
\end{lemma}
\noindent\textbf{Proof of Lemma \ref{maxminpath}}:
I prove by induction.

\noindent\textbf{Step 1: }When \(S_{i_1,\cdots,i_k} \in \mathcal{T}\), the equality holds true by definition of \(\kappa\).

\noindent\textbf{Step 2: }Consider any \(S_{i_1,\cdots,i_{k},l}\) with \(l=1,\cdots,N_{i_1,\cdots,i_k}\). Suppose the equality in Lemma \ref{maxminpath} holds true for any such \(l\). I show below that the equality also holds true for \(S_{i_1,\cdots,i_k}\). 
\begin{align*}
    &\max_{S_{j_1,\cdots,j_n} \in \mathcal{T}\cap\mathcal{D}_{i_1,\cdots,i_k}} \min \left\{\sigma_{j_1},\sigma_{j_2},\cdots,\sigma_{j_1,\cdots,j_n},\sigma_{j_1,\cdots,j_n,1}\right\}\\
    =& \max_{1 \leq l \leq N_{i_1,\cdots,i_k}} \max_{S_{j_1,\cdots,j_n \in \mathcal{T}\cap \mathcal{D}_{i_1,\cdots,i_k,l}}} \min \left\{\sigma_{j_1},\sigma_{j_2},\cdots,\sigma_{j_1,\cdots,j_n},\sigma_{j_1,\cdots,j_n,1}\right\}\\
    =&\max_{1 \leq l \leq N_{i_1,\cdots,i_k}} \min \left\{\sigma_{i_1},\cdots,\sigma_{i_1,\cdots,i_k},\kappa_{i_1,\cdots,i_k,l} \right\} \\
    =&\max_{1 \leq l \leq N_{i_1,\cdots,i_k}} \left[\min \left\{\sigma_{i_1},\cdots,\sigma_{i_1,\cdots,i_k}\right\}-\Biggl(\min \left\{\sigma_{i_1},\cdots,\sigma_{i_1,\cdots,i_k}\right\}-\kappa_{i_1,\cdots,i_k,l}\Biggr)^{+}\right]\\
    =&\min \left\{\sigma_{i_1},\cdots,\sigma_{i_1,\cdots,i_k}\right\}-\Biggl(\min \left\{\sigma_{i_1},\cdots,\sigma_{i_1,\cdots,i_k}\right\}-\max_{1 \leq l \leq N_{i_1,\cdots,i_k}} \kappa_{i_1,\cdots,i_k,l}\Biggr)^{+}\\
    =&\min \left\{\sigma_{i_1},\cdots,\sigma_{i_1,\cdots,i_k},\max_{1 \leq l \leq N_{i_1,\cdots,i_k}} \kappa_{i_1,\cdots,i_k,l}\right\}\\
    =&\min \left\{\sigma_{i_1},\cdots,\sigma_{i_1,\cdots,i_{k-1}},\kappa_{i_1,\cdots,i_k}\right\}.
\end{align*}
\noindent\textbf{Step 3: }By induction, the equality holds true for any \(i_1,\cdots,i_k\). This finishes the proof. \qed

\bigskip

In order for the equality to fail, there must exist \(i_1,\cdots,i_k\), \(l'\) and \(l''\) such that \(\kappa_{i_1,\cdots,i_k,l'} > \kappa_{i_1,\cdots,i_k,l''}\), \(\sum_{S_{j_1,\cdots,j_n} \in \mathcal{T}\cap\mathcal{D}_{i_1,\cdots,i_k,l'}}\mathbb{A}_{j_1,\cdots,j_n}=0\) but \(\sum_{S_{j_1,\cdots,j_n} \in \mathcal{T}\cap\mathcal{D}_{i_1,\cdots,i_k,l''}}\mathbb{A}_{j_1,\cdots,j_n}=1\). Suppose \(\mathbb{A}_{j^*_1,\cdots,j^*_{n}}=1\), where \(S_{j^*_1,\cdots,j^*_{n}} \in \mathcal{T}\cap\mathcal{D}_{i_1,\cdots,i_k,l''}\). This implies
\begin{align*}
\min \{\sigma_{j^*_1},\cdots,\sigma_{j^*_1,\cdots,j^*_n},\sigma_{j^*_1,\cdots,j^*_n,1}\} \geq \max_{S_{j_1,\cdots,j_n} \in \mathcal{T}\cap\mathcal{D}_{i_1,\cdots,i_k,l'}} \min \left\{\sigma_{j_1},\sigma_{j_2},\cdots,\sigma_{j_1,\cdots,j_n},\sigma_{j_1,\cdots,j_n,1}\right\}.  
\end{align*}
Furthermore, by Lemma \ref{maxminpath},
\begin{align*}
    &\kappa_{i_1,\cdots,i_k,l''} \\
    \geq &\min \{\sigma_{j^*_1},\cdots,\sigma_{j^*_1,\cdots,j^*_n},\sigma_{j^*_1,\cdots,j^*_n,1}\} \\
    \geq &\max_{S_{j_1,\cdots,j_n} \in \mathcal{T}\cap\mathcal{D}_{i_1,\cdots,i_k,l'}} \min \left\{\sigma_{j_1},\sigma_{j_2},\cdots,\sigma_{j_1,\cdots,j_n},\sigma_{j_1,\cdots,j_n,1}\right\}\\
    =&\kappa_{i_1,\cdots,i_k,l'},
\end{align*}
a contradiction. This finishes the proof.\qed

\bigskip

\noindent\textbf{Proof of Lemma \ref{maxbound}}: To show the index policy maximizes the upper-bound, I only need to show that 
\begin{align*}
    \kappa_{i'_1}>\kappa_{i''_1} \implies \sum_{S_{j_1,\cdots,j_n} \in \mathcal{T}\cap\mathcal{D}_{i'_1}} \mathbb{A}_{j_1,\cdots,j_n} \geq \sum_{S_{j_1,\cdots,j_n} \in \mathcal{T}\cap\mathcal{D}_{i''_1}} \mathbb{A}_{j_1,\cdots,j_n}.
\end{align*}
Suppose to the contrary that \(\kappa_{i'_1}>\kappa_{i''_2}\), \(\sum_{S_{j_1,\cdots,j_n} \in \mathcal{T}\cap\mathcal{D}_{i'_1}} \mathbb{A}_{j_1,\cdots,j_n}=0\) but \(\sum_{S_{j_1,\cdots,j_n} \in \mathcal{T}\cap\mathcal{D}_{i''_1}} \mathbb{A}_{j_1,\cdots,j_n}=1\). Apply Lemma \ref{maxminpath} to the child nodes of the root, then 
\begin{align*}
    \kappa_{i_1}=\max_{S_{j_1,\cdots,j_n} \in \mathcal{T} \cap \mathcal{D}_{i_1}} \min \left\{\sigma_{j_1},\cdots,\sigma_{j_1,\cdots,j_n},\sigma_{j_1,\cdots,j_n,1}\right\}.
\end{align*}
Since \(\sum_{S_{j_1,\cdots,j_n} \in \mathcal{T}\cap\mathcal{D}_{i''_1}} \mathbb{A}_{j_1,\cdots,j_n}=1\), there must exist \(S_{j^*_1,\cdots,j^*_n} \in \mathcal{T}\cap\mathcal{D}_{i''_1}\) such that \(\mathbb{A}_{j^*_1,\cdots,j^*_n}=1\). This further implies 
\begin{align*}
\min \left\{\sigma_{j^*_1},\cdots,\sigma_{j^*_1,\cdots,j^*_n},\sigma_{j^*_1,\cdots,j^*_n,1}\right\} \geq \max_{S_{j_1,\cdots,j_n} \in \mathcal{T} \cap \mathcal{D}_{i'_1}} \min \left\{\sigma_{j_1},\cdots,\sigma_{j_1,\cdots,j_n},\sigma_{j_1,\cdots,j_n,1}\right\}.
\end{align*}
However, this means \(\kappa_{i''_1} \geq \kappa_{i'_1}\), a contradiction. This finishes the proof.\qed

\bigskip

\noindent\textbf{Proof of Lemma \ref{inchazard}: }See the proof of the first part of Proposition 2 in \cite{choi2018consumer}. \qed

\bigskip

\noindent\textbf{Proof of Lemma \ref{smallcost}: }This is easily obtained by using the monotonicity of \(w\) and \(\frac{1-H(w)}{h(w)}\). \qed

\bigskip

\noindent\textbf{Proof of Proposition \ref{equi1}: } The equilibrium candidate is calculated in the main text. It remains to verify that \(p^*\) is indeed an equilibrium with active consumer search. Assumption \ref{logconcave} makes sure that the profit function first decreases in \(p\) and then increases in \(p\), and thus FOC is sufficient for optimality. Also, \(\sigma>\sigma^*\) implies \(\sigma>p^*=\frac{1-H(\sigma)}{h(\sigma)}\), so there is active search.  \qed

\bigskip

\noindent\textbf{Proof of Proposition \ref{equi2}: }The equilibrium candidate is calculated in the main text. It remains to verify that \(p^{\dagger}\) is indeed an equilibrium with active consumer search. I show that \(\frac{\int_{\sigma-r}^{\infty}[1-G(w-p^{\dagger}+p-x)]dF(x)}{\int_{\sigma-r}^{\infty}g(w-p^{\dagger}+p-x)dF(x)}\) is decreasing in \(p\), and thus FOC is sufficient for optimality. To see this, define a random variable \(\tilde{X}\) supported over \([\sigma-r,\infty)\) as follows: Let its distribution \(\tilde{F}(.):=\frac{F(.)}{1-F(\sigma-r)}\). Then one can see that
\begin{align*}
    \frac{\int_{\sigma-r}^{\infty}[1-G(\sigma-p^{\dagger}+p-x)]dF(x)}{\int_{\sigma-r}^{\infty}g(\sigma-p^{\dagger}+p-x)dF(x)}=\frac{\int_{-\infty}^{\infty}[1-G(\sigma-p^{\dagger}+p-x)]d\tilde{F}(x)}{\int_{-\infty}^{\infty}g(\sigma-p^{\dagger}+p-x)d\tilde{F}(x)},
\end{align*}
which is the inverse hazard rate of the random variable \(\tilde{X}+Y\), evaluated at \(\sigma-p^{\dagger}+p\). Since \(X\) and \(Y\)  both have log concave distributions, the distribution of \(\tilde{X}+Y\) is also log concave, and thus the inverse hazard rate is decreasing. Then by increasing \(p\), the whole term decreases. Additionally, \(p^{\dagger} \leq \frac{1-G(r)}{g(r)} <\sigma\), so there is active consumer search.
\qed

\bigskip

\noindent\textbf{Proof of Proposition \ref{stage2higher}: }Note that, \(W=X+\min \{Y,r\}> \sigma\) is equivalent to \(X+Y > \sigma\) and \(X+r>\sigma\). Since \(X\) and \(Y\) are continuous random variables, 
\begin{align*}
\mathrm{Pr}(W > \sigma)&=1-H(\sigma)\\
&=\int_{\sigma-r}^{\infty} [1-G(\sigma-x)]dF(x),
\end{align*}
which is exactly the numerator of \(p^{\dagger}\). Furthermore,
\begin{align*}
h(\sigma)&=-\frac{d\int_{\sigma-r}^{\infty} [1-G(\sigma-x)]dF(x)}{d\sigma}\\
&=\int_{\sigma-r}^{\infty} g(\sigma-x)dF(x)+[1-G(r)]f(\sigma-r)\\
&\geq \int_{\sigma-r}^{\infty} g(\sigma-x)dF(x),
\end{align*}
with equality when \(\underline{x}+r \geq \sigma\), i.e., there is no option value.
This finishes the proof. \qed

\bigskip

\noindent\textbf{Proof of Lemma \ref{option}: } Observe that
\begin{align*}
    \mathbb{E}[(X+\min \{Y,r\}-\sigma)^++(Y-r)^+]=c_X+c_Y
\end{align*}
and
\begin{align*}
    (X+\min \{Y,r\}-\sigma)^++(Y-r)^+ \geq (X+Y-\sigma)^+.
\end{align*}
Therefore
\begin{align*}
    \mathbb{E}[(X+Y-\sigma)^+] \leq c_X+c_Y,
\end{align*}
which gives \(\sigma \geq \hat{\sigma}\).  \qed

\bigskip

\noindent\textbf{Proof of Lemma \ref{elasticity}: } Let \(\lambda\) denote the hazard rate. Define \(S(t)=1-F(t)\).  For \(\underline{x}+\underline{y} \leq t \leq \overline{x}+r\),
\begin{align*}
    &\lambda_W(t)=\frac{\int_{-\infty}^{\infty} f(t-\min \{y,r\})g(y)dy}{\int_{-\infty}^{\infty} S(t-\min\{y,r\})g(y)dy}=\frac{\int_{-{\infty}}^rf(t-y)g(y)dy+\int_r^{\infty}f(t-r)g(y)dy}{\int_{-\infty}^rS(t-y)g(y)dy+\int_r^{\infty}S(t-r)g(y)dy}=\frac{A+a}{B+b},\\
    &\lambda_Z(t)=\frac{\int_{-\infty}^{\infty} f(t-y)g(y)dy}{\int_{-\infty}^{\infty} S(t-y)g(y)dy}=\frac{\int_{-{\infty}}^rf(t-y)g(y)dy+\int_r^{\infty}f(t-y)g(y)dy}{\int_{-\infty}^rS(t-y)g(y)dy+\int_r^{\infty}S(t-y)g(y)dy}=\frac{A+c}{B+d}.
\end{align*}
Notice that when \(t \leq \underline{x}+r\), \(a=c\) and \(b=d\).
Note that \(\frac{f(x)}{S(x)}\) is increasing for \(-\infty <x <\overline{x}\), but I need to deal with the end points of the integrals since \(\frac{f(x)}{S(x)}\) is not well defined for \(x \geq \overline{x}\).

Since \(F\) has increasing hazard rate,
\(\frac{c}{d} \leq \frac{f(t-r)}{S(t-r)}= \frac{a}{b} \leq \frac{A}{B}\). Also, \(b \leq d\) since \(S\) is decreasing. Let \(t_1=\frac{c}{d}, t_2=\frac{a}{b},t_3=\frac{A}{B}\). Then,
\begin{align*}
    \lambda_W(t)-\lambda_Z(t)&=\frac{A+a}{B+b}-\frac{A+c}{B+d}\\
    &=\frac{1}{(B+b)(B+d)}(AB+Ad+aB+ad-AB-Ab-Bc-bc)\\
    &=\frac{1}{(B+b)(B+d)}(t_3 Bd+t_2Bb+t_2bd-t_3 Bb-t_1Bd-t_1bd)\\
    &=\frac{1}{(B+b)(B+d)}[(t_3-t_1)Bd+(t_2-t_1)bd+(t_2-t_3)Bb]\\
    &\geq \frac{1}{(B+b)(B+d)}[(t_3-t_1)Bb+(t_2-t_1)bd+(t_2-t_3)Bb]\\
    &=\frac{1}{(B+b)(B+d)}[(t_2-t_1)(bd+Bb)]\\
    &\geq 0.
\end{align*} 
This finishes the proof. \qed

\bigskip

\noindent\textbf{Proof of Theorem \ref{lowerprice}: } By Lemma \ref{option} and \ref{elasticity}, \(p^*=\frac{1}{\lambda_W(\sigma)} \leq \frac{1}{\lambda_Z(\sigma)} \leq \frac{1}{\lambda_Z(\hat{\sigma})}=\hat{p}\).\qed

\bigskip

\noindent\textbf{Proof of Corollary \ref{highercs}: } According to Lemma \ref{upperbound}-\ref{maxbound}, the buyer's expected payoff in the benchmark case is
\begin{align*}
    \mathbb{E}\left[\sum_{i=1}^n \mathbb{A}_i \min \{\hat{\sigma}-\hat{p},X_i+Y_i-\hat{p}\}\right].
\end{align*}
Using the sandwich theorem,
\begin{align*}
    \lim_{n \to \infty} [1-K(\hat{\sigma})^n](\hat{\sigma}-\hat{p})\leq \lim_{n \to \infty} \mathbb{E}\left[\sum_{i=1}^n \mathbb{A}_i \min \{\hat{\sigma}-\hat{p},X_i+Y_i-\hat{p}\}\right] \leq \hat{\sigma}-\hat{p},
\end{align*}
which implies \(\mathbb{E}\left[\sum_{i=1}^n \mathbb{A}_i \min \{\hat{\sigma}-\hat{p},X_i+Y_i-\hat{p}\}\right]=\hat{\sigma}-\hat{p}\). Similarly, the buyer's expected payoff under gradual resolution of uncertainty is \(\sigma-p^*\). Since \(\sigma \geq \hat{\sigma}\) and \(p^* \leq \hat{p}\), the consumer surplus is higher.
\qed

\bigskip

\noindent\textbf{Proof of Lemma \ref{nomix2}: }Given any equilibrium, take an arbitrary \(p_1\) set by firm \(i\). I show that there is a unique stage 2 price that maximizes the firm's profit. To see this, let \(\sigma-\tilde{p}\) be the stage 1 index. Having inspected \(X_i\), buyers see \(p_1\) and believe that the stage 2 price will be drawn according to the putative equilibrium and will use the stage 2 index \(X_i+r-p\) to guide their search. If firm \(i\) sets \(p'\), \(\kappa_i=\min \{\sigma-\tilde{p},X_i+r-p,X_i+Y_i-p'\}\). Firm \(i\)'s demand under monopolistic competition is thus a constant times \(\mathrm{Pr}(\kappa_i \geq \sigma-\tilde{p})=\int_{\sigma-\tilde{p}+p-r}^{\infty}[1-G(\sigma-\tilde{p}+p'-x)]dF(x)\). The optimal \(p'\) should solve the FOC:
\begin{align*}
    \frac{d\mathrm{Pr}(\kappa_i \geq \sigma-\tilde{p})p'}{dp'}=p'\int_{\sigma-\tilde{p}+p-r}^{\infty}-g(\sigma-\tilde{p}+p'-x)]dF(x)+\int_{\sigma-\tilde{p}+p-r}^{\infty}[1-G(\sigma-\tilde{p}+p'-x)]dF(x)=0,
\end{align*}
which gives 
\begin{align*}
    p'=\frac{\int_{\sigma-\tilde{p}+p-r}^{\infty}[1-G(\sigma-\tilde{p}+p'-x)]dF(x)}{\int_{\sigma-\tilde{p}+p-r}^{\infty}g(\sigma-\tilde{p}+p'-x)]dF(x)}.
\end{align*}
Since the RHS decreases with \(p'\), the FOC is sufficient and there is a unique \(p'\) that solves it. This finishes the proof. \qed

\bigskip

\noindent\textbf{Proof of Lemma \ref{uniquestage2}: }Suppose that in an equilibrium, firm \(i\) sets \(p\) at stage 2. Then given \(\kappa_i=\min \{\sigma-\tilde{p},X_i+r-p,X_i+Y_i-p'\}\), \(p\) has to be such that setting \(p'=p\) is optimal. Building on the last proof, this requires 
\begin{align*}
    p&=\frac{\int_{\sigma-\tilde{p}+p-r}^{\infty}[1-G(\sigma-\tilde{p}+p-x)]dF(x)}{\int_{\sigma-\tilde{p}+p-r}^{\infty}g(\sigma-\tilde{p}+p-x)]dF(x)}\\
    &=\frac{\int_{-\infty}^r [1-G(y)]f(\sigma-\tilde{p}+p-y)dy}{\int_{-\infty}^r g(y)f(\sigma-\tilde{p}+p-y)dy}.
\end{align*}
To prove this lemma, I show the RHS is decreasing, which gives the uniqueness of \(p\). This is equivalent to proving the following lemma:
\begin{lemma}\label{Karlin}
    \(\frac{\int_{-\infty}^r [1-G(y)]f(t-y)dy}{\int_{-\infty}^r g(y)f(t-y)dy}\) is decreasing in \(t\).
\end{lemma}
\noindent\textbf{Proof of Lemma \ref{Karlin}: } Take an arbitrary \(c\) and define \(N_c(t)=\int_{-\infty}^r \phi_c(t)K(t,y)dy\), where \(K(t,y)=f(t-y)\) and \(\phi_c(t)=1-G(y)-cg(y)\). Since \(F\) and \(G\) are log-concave distributions, \(K\) is \(TP_2\), \(\phi_c\) is single-crossing and it goes from positive to negative. By Theorem 3.1 in \cite{karlin1968total}, \(N_c\) is also single-crossing and it exhibits the same sequence of sign changes as \(\phi_c\). Note also that,
\begin{align*}
    N_c(t)=\left[\frac{\int_{-\infty}^r [1-G(y)]f(t-y)dy}{\int_{-\infty}^r g(y)f(t-y)dy}-c\right]\int_{-\infty}^r g(y)f(t-y)dy,
\end{align*}
so \(\frac{\int_{-\infty}^r [1-G(y)]f(t-y)dy}{\int_{-\infty}^r g(y)f(t-y)dy}-c\) has the same sign as \(N_c(t)\). This establishes the lemma.
\qed

\bigskip

Thus, given any \(\tilde{p}\), there exists a unique \(p\) under which the firm has no incentive to deviate. This finishes the proof. \qed

\bigskip

\noindent\textbf{Proof of Proposition \ref{unique2}: }I first show that \(p^{\dagger}\) is the only candidate for equilibrium, then I show that it is indeed an equilibrium.

\noindent\textbf{Part 1: }According to Lemma \ref{nomix2} and \ref{uniquestage2}, \(\kappa_i=\{\sigma-p,X_i+r-p,X_i+Y_i-p\}\). Analysis in Section \ref{s43} shows that the only candidate is \(p=p^{\dagger}\).

\noindent\textbf{Part 2: }To show that \(p=p^{\dagger}\) is indeed an equilibrium, note that only deviating at stage 2 is not profitable from Section \ref{s43}. For deviation at stage 1, it deters all consumer search, and thus cannot be profitable. \qed

\bigskip

\noindent\textbf{Proof of Lemma \ref{nomix1}: }Similar to the proof of Lemma \ref{nomix2}, firm \(i\)'s demand is a constant times \(\mathrm{Pr}(\kappa_i \geq \sigma-\tilde{p})=\int_{\sigma-\tilde{p}+p'-r}^{\infty}[1-G(\sigma-\tilde{p}+p-x)]dF(x)\), thus taking the derivative of the profit (ignoring this constant) yields
\begin{align*}
     &p\int_{\sigma-\tilde{p}+p'-r}^{\infty}-g(\sigma-\tilde{p}+p-x)]dF(x)+\int_{\sigma-\tilde{p}+p'-r}^{\infty}[1-G(\sigma-\tilde{p}+p-x)]dF(x)\\
     =&\left[\frac{\int_{\sigma-\tilde{p}+p'-r}^{\infty}[1-G(\sigma-\tilde{p}+p-x)]dF(x)}{\int_{\sigma-\tilde{p}+p'-r}^{\infty}g(\sigma-\tilde{p}+p-x)]dF(x)}-p\right]\int_{\sigma-\tilde{p}+p'-r}^{\infty}g(\sigma-\tilde{p}+p-x)]dF(x),
\end{align*}
so the profit first increases in \(p\); after the peak, it decreases in \(p\). Thus, there is a unique \(p\) that maximizes the firm's profit under the constraint that the price cannot exceed the price at stage 1. This finishes the proof. \qed

\bigskip

\noindent\textbf{Proof of Lemma \ref{uniquestage}: }Given a putative equilibrium, suppose that there exists some stage 1 price such that buyers believe the price revealed at stage 2 will be \(p\) upon seeing this price. Note that there can be a range of \(p\) such that the firm has no profitable deviation under the regulation. Under such \(p\), the optimal \(p'\) without regulation is  greater than \(p\), so under the regulation the optimal price at stage 2 is \(p\) itself. Such \(p\) is determined by the following inequality:
\begin{align*}
    p-\frac{\int_{\sigma-\tilde{p}+p-r}^{\infty}[1-G(\sigma-\tilde{p}+p-x)]dF(x)}{\int_{\sigma-\tilde{p}+p-r}^{\infty}g(\sigma-\tilde{p}+p-x)]dF(x)} \leq 0.
\end{align*}
Since the LHS is strictly increasing, given \(\tilde{p}\), this range of \(p\) forms an interval. Next, consider \(\kappa_i=\min \{\sigma-\tilde{p},X_i+r-p,X_i+Y_i-p\}\). The profit for firm \(i\) at \(p\) is 
\begin{align*}
    p[1-H(\sigma-\tilde{p}+p)],
\end{align*}
and taking its derivative yields
\begin{align*}
    \left[\frac{1-H(\sigma-\tilde{p}+p)}{h(\sigma-\tilde{p}+p)}-p\right]h(\sigma-\tilde{p}+p),
\end{align*}
so the profit first increases in \(p\); after the peak, it decreases in \(p\). Therefore, if an equilibrium reveals multiple distinct prices at stage 2, these prices must fall on opposite sides of the peak. However, this cannot be an equilibrium phenomenon as the firm can profitably deviate by setting prices at both stages at the peak. Observing such a deviation, the buyer will think the price stays the same at stage 2 and will follow the optimal search policy given this belief. The firm's optimal price given buyers' search strategy is exactly the same price, which gives a strictly higher profit. This finishes the proof.
\qed

\bigskip

\noindent\textbf{Proof of Proposition \ref{unique1}: }I first show that \(p^*\) is the only candidate for equilibrium, then I show that it is indeed an equilibrium.

\noindent\textbf{Part 1: }According to Lemma \ref{nomix1} and \ref{uniquestage}, a unique price is revealed at stage 2, the price at stage 1 exactly equals this price and it is exactly at the peak. This gives
\begin{align*}
    p=\frac{1-H(\sigma)}{h(\sigma)},
\end{align*}
which is exactly \(p^*\).

\noindent\textbf{Part 2: }I show that there is no profitable deviation with \(p_1=p_2=p^*\). There are two types of deviation: revealing \(p^*\) at stage 1 and another price at stage 2, or revealing a different price at stage 1. For the first type of deviation, I show that under \(p^*\), the optimal price at stage 2 without the regulation is greater than \(p^*\). Since \(\kappa_i=\min \{\sigma-p^*,X_i+r-p^*,X_i+Y_i-p'\}\), the profit is \(p'\int_{\sigma-r}^{\infty}[1-G(\sigma-p^*+p'-x)]dF(x)\), and the FOC is 
\begin{align*}
    p'=\frac{\int_{\sigma-r}^{\infty}[1-G(\sigma-p^*+p'-x)]dF(x)}{\int_{\sigma-r}^{\infty}g(\sigma-p^*+p'-x)dF(x)}.
\end{align*}
When \(p'=p^*\), the RHS is exactly \(p^{\dagger}\), which is greater than \(p^*\) by Proposition \ref{stage2higher}, thus the solution \(p' \geq p^*\), implying there is no profitable deviation. For the second type of deviation, suppose the firm sets stage 1 price \(\overline{p} \neq p^*\). Since this deviation is off-path, according to the belief updating rule, buyers believe that \(\overline{p}\) will also be the price charged at stage 2. 
However, in this case the firm’s profit coincides with that from deviating to \(p=\overline{p}\) in Section \ref{s42}. And since \(p^*\) is indeed an equilibrium there, this deviation cannot be profitable. This establishes the result.
\qed

\bigskip
\noindent\textbf{Proof of Proposition \ref{unique2weaker}: }   Following the analysis in the main text, I only check there is no profitable stage 1 deviation. Suppose firm \(i\) deviates at stage 1. By Assumption \ref{weaker}, consumers believe the final price will be higher than \(p^{\dagger}\), say it is \(\overline{p}\). Firm \(i\)'s profit then satisfies
\begin{align*}
    &\max_{\tilde{p}}\tilde{p}\int_{\sigma-p^{\dagger}+\overline{p}-r}^{\infty} [1-G(\sigma-p^{\dagger}+\tilde{p}-x)]dF(x)\\
    \leq &\max_{\tilde{p}}\tilde{p}\int_{\sigma-r}^{\infty} [1-G(\sigma-p^{\dagger}+\tilde{p}-x)]dF(x),
\end{align*}
which is the profit before deviation. This finishes the proof.
\qed

\end{document}